\documentclass[12pt,a4paper,oneside]{amsart}

\usepackage{hyperref}
\usepackage[utf8]{inputenc}
\usepackage[english]{babel}
\usepackage{graphicx}
\usepackage{multirow}
\usepackage{caption}
\usepackage{amssymb,amsmath,amsthm}
\usepackage{vmargin}
\usepackage{xcolor}
\usepackage{mathrsfs}
\usepackage{subfigure}
\usepackage{pstricks}
\usepackage{pstricks-add}
\usepackage{multido}
\usepackage{pst-fill}
\usepackage{pifont}

\usepackage{color}

\theoremstyle{plain}
\newtheorem{theorem}{Theorem}
\newtheorem{lemma}[theorem]{Lemma}
\newtheorem{corollary}[theorem]{Corollary}

\theoremstyle{definition}

\theoremstyle{remark}
\newtheorem{remark}{Remark}

\numberwithin{equation}{section}
\numberwithin{theorem}{section}

\newcommand{\bX}{\mathbf X}

\newcommand{\bY}{\mathbf Y}
\newcommand{\bZ}{\mathbf Z}
\newcommand{\bW}{\mathbf W}
\newcommand{\bS}{\mathbf S}
\newcommand{\bD}{\mathbf D}
\newcommand{\be}{\mathbf\epsilon}

\newcommand{\R}{\mathbb R}
\newcommand{\N}{\mathbb N}
\newcommand{\E}{\mathbb E}
\renewcommand{\P}{\mathbb P}
\newcommand{\1}{\mathbf 1}
\newcommand{\T}{\mathscr T}
\newcommand{\Nor}{\mathcal N}
\newcommand{\B}{\mathcal B}

\begin{document}

\title[Change point detection for high-dimensional
regression]{Computationally efficient change point detection for
    high-dimensional regression}  

\author{Florencia Leonardi}
\address{University of S\~ao Paulo, Brazil}
\email{florencia@usp.br}

\author{Peter B\"uhlmann}
\address{Seminar for Statistics, ETH Z\"urich, Switzerland}
\email{buhlmann@stat.math.ethz.ch}

\date{\today}

\thanks{This article was produced as part of the activities of FAPESP  Research, Innovation and Dissemination Center for Neuromathematics (grant \#2013/ 07699-0 , S.Paulo Research Foundation). F.L. was partially supported by a FAPESP's fellowship (grant \#2014/00947-0) and CNPq's fellowship (grant \#233216/2014-6).}

\maketitle

\begin{abstract}
Large-scale sequential data is often exposed to some degree of
  inhomogeneity in the form of sudden changes in the parameters of the
data-generating process.
We consider the problem of detecting such structural changes
in a high-dimensional regression setting. We propose a joint
estimator of the number and the locations of the change points and of the
parameters in the corresponding segments. The estimator can be computed
using dynamic programming or, as we emphasize here, it can be approximated
using a binary search algorithm with $O(n \log(n) \mathrm{Lasso}(n))$
computational operations while still enjoying essentially the same
theoretical properties; here $\mathrm{Lasso}(n)$ 
denotes the computational cost of computing the Lasso for sample size $n$.   
We establish oracle inequalities for the estimator as well as for its
binary search approximation, covering also the case with a large 
(asymptotically growing) number of change points. We evaluate the
performance of the proposed estimation algorithms on simulated data and
apply the methodology to real data.   
\end{abstract}

\section{Introduction}

Much progress and work has been done in the last decade on methodology and
theory of high-dimensional data, and we refer to
\cite{fanlv10,peter-sara-book} for 
some overview. The vast majority of the focus has been on regression or
classification with homogeneous data from a model with the same
high-dimensional parameter for all the 
samples. Such a homogeneity assumption is not realistic for some datasets,
in particular for large-scale data where sample size and the
dimensionality are large. Some work addressing the issue of heterogeneous
data in high-dimensional settings include factor models
\cite{mclachlan2003modelling,carvalho2008high,fan2008high}, mixture
regression models \cite{staedleretal10}, change point regression models
\cite{leeetal15} or ``maximin'' worst case analysis \cite{mebu14}. 

We consider here a change point, high-dimensional regression model. We
propose a joint estimator, using regularization with $\ell_1$-norms of the
parameters in different segments, for the
number and the locations of the change points and for the parameters of
each corresponding segment. We establish
an oracle inequality and consistency for the number of change points,
implying near optimal convergence rates for the underlying regression
parameters. Our analysis includes the case where the number of change
points can be large (and asymptotically growing). 

Our estimator can be computed using dynamic programming. To markedly speed up
computational time for large-scale data, we can use a computationally
efficient binary search algorithm, having computational cost of the order
$O(n \log(n) \mathrm{Lasso}(n))$, to approximate the estimator \cite{killick2012,fryzlewicz2014}; here
$\mathrm{Lasso}(n)$ denotes the cost to compute the Lasso for sample size
$n$. 
A main result of our paper establishes that the binary search
algorithm essentially enjoys the same theoretical properties as the
original estimator. We thus provide a strong justification for using binary
search in change point detection in large-scale regression problems.  

We evaluate the performance of the estimation algorithms by means of
simulations and we also show the utility of our  
approach for real data.
Our work is related to the one in \cite{leeetal15} and we will
outline the differences in Section \ref{subsec.relatedwork}. 

The problem of change point detection has been 
studied already by e.g. Page \cite{page1955} and since the early 1980s there has
been an explosion of contributions (see \cite{review2013} and references
therein). Change point models cover a wide range of applications, from
e.g. econometrics \cite{andreou2002,cho-fryzlewicz2015} to genomics
\cite{braun2000,boys2004,decastro2015}. In most of the literature, ``change
point detection'' deals with the problem 
of finding the piecewise constant means in univariate or multivariate data,
see for example \cite{fricketal14} which contains many references. There
are also some  works studying changes in the parameters of autoregressive models
\cite{chan2014} or on network data \cite{kolar2012,barnett2014}. A vast list of
contributions on the change point detection problem can be found in 
the recent review paper \cite{review2013} or in the repository
\cite{changepoint}. However, change point models for
high-dimensional regression or classification where the number of
parameters can be much larger than sample size have not been considered
very much.  

\subsection{Related work and our contribution}\label{subsec.relatedwork}

We propose a joint estimator of 
the change points and the parameters for each segment in a
high-dimensional linear model, even in the case where 
the number of segments is unknown. 
To the best of our knowledge, there is only the independently developed
work \cite{leeetal15} which is related to our study in the sense of
considering a similar motivation and high-dimensional model. In
\cite{kolar2012}, an undirected Gaussian network model is considered which
can be broken into single regressions: the mathematical analysis is not
treating the high-dimensional case, and the proposed approach is based on a
total-variation, Fused Lasso type penalty with a corresponding approximate
computational optimization only. As described next, our results cover multiple,
high-dimensional change point regression models with corresponding
theoretical guarantees of a computationally efficient algorithm. 

In \cite{leeetal15}, a high-dimensional linear model with
\emph{one} potential change point for two different high-dimensional
regression parameters is considered. We address here the situation with
multiple change points, with a possibly growing number thereof as the
sample size increases. In particular, we face here also the issue if
efficient computation (as mentioned in the next paragraph) as well as the
problem of determining the 
number of change points. We use the Lasso, similarly as in
\cite{leeetal15}, for each segment arising from the change points and we
then minimize an overall penalized residual sum of squares.
We prove an oracle
inequality for the penalized residual sum of squares procedure using a sum
of $\ell_1$-norm penalties. The result implies near optimal convergence
rates for the parameters and in addition, we obtain directly a consistent
estimator for the possibly growing number of change points, without the
need to do some additional model selection in the spirit of e.g. BIC.  

Furthermore and especially important for large-scale data, and since we are
considering multiple and possibly very many change points, we focus on the
computational task as well whereas the case 
with one change point as in \cite{leeetal15} is computationally very
easy. While a dynamic 
programming algorithm works in general, we prove that a much more
efficient binary search algorithm is
consistent and has (essentially) the 
same rates in the oracle inequality as mentioned above. Of course, binary
segmentation algorithms are not new, see for example \cite{fryzlewicz2014},
but the derivation of a theoretical consistency guarantee in the
high-dimensional change point problem as considered here is entirely
novel. 

\section{Change point model and estimation}

Consider a sequence of independent observations $\{(Y_i,X_i)\}_{i=1}^n$
with $p$-dimensional covariates $X_i\in\R^p$ and univariate response
$Y_i\in\R$. Assume $\{X_i\}_{i=1}^n$ are i.id. with covariance matrix $\Sigma$
and $\{Y_i\}_{i=1}^n$ are given by   
\begin{equation}\label{model}
Y_i = X_i^T\beta^{(i)} + \epsilon_i\ \ (i=1,\ldots,n),
\end{equation}
where $\epsilon_1,\dotsc,\epsilon_n$ are i.i.d., independent of $X_1,\dotsc,X_n$, and 
$\{\beta^{(i)}\}_{i=1}^n$ are piecewise  constant. The i.i.d. assumption for
$\{X_i\}_{i=1}^n$ is made for notational simplicity but can be relaxed to
an i.i.d. assumption within each segment where $\{\beta^{(i)}\}_{i=1}^n$ is
constant.
That is, we assume 
 there exists a $(k_0+1)$-dimensional vector $\alpha^0=(\alpha^0_0,\dotsc,\alpha^0_{k_0})$ satisfying
\begin{equation}\label{alph0}
0=\alpha^0_0 < \alpha^0_1<\dotsc < \alpha^0_{k_0}=1
\end{equation}
and $k_0$ real vectors $\beta^0(1),\dotsc, \beta^0(k_0)$ in $\R ^p$ such that
\begin{equation}\label{beta0}
\beta^{(i)} \;=\; \sum_{j=1}^{k_0} \beta^0(j)\1\{i/n \in
(\alpha^0_{j-1},\alpha^0_j] \}\,.\end{equation}
This means that the sequence $(X_1,Y_1),\dotsc, (X_n,Y_n)$ is independent but
only piecewise identically distributed, with change points at the elements
of $\alpha^0$. To simplify notation here and in the sequel we assume,
without loss of generality,  
that $\alpha^0_j n \in \N$ for all $j=1,\dotsc,k_0$.  

Sometimes we will use matrix notation for the equations in
\eqref{model}. Given an interval $(u,v]\subset[0,1]$ such that $un,vn\in\N$
we will denote by $\bY_{(u,v]}$ the  vector $(Y_{un+1},\dotsc,Y_{vn})^T$
and by $\be_{(u,v]}$ the  vector
$(\epsilon_{un+1},\dotsc,\epsilon_{vn})^T$. Analogously, $\bX_{(u,v]}$ will
denote the $(v-u)n\times p$ matrix
$(X^{(1)}_{(u,v]},\dotsc,X^{(p)}_{(u,v]})$. Then the model in
\eqref{model} can be written as 
\begin{equation}\label{eqmat}
\bY_{(\alpha^0_{j-1},\alpha^0_j]} = \bX_{(\alpha^0_{j-1},\alpha^0_j ]} \beta^0(j) +\epsilon_{(\alpha^0_{j-1},\alpha^0_j  ]}
\end{equation}
for $j=1,\dots,k_0$.

We propose a joint estimator for the change points and the regression
parameters in  the model given by  \eqref{model},  \eqref{alph0} and
\eqref{beta0}, without assuming a known upper bound on the number of
segments.  
Given a vector $\alpha=(\alpha_0,\dotsc,\alpha_k)$ satisfying
\begin{equation}\label{alph}
0=\alpha_0 < \alpha_1<\dotsc < \alpha_k=1
\end{equation}
we  denote by $\ell(\alpha)$ the number of positive components, that is 
$\ell(\alpha) = k$. This value also corresponds to the number of segments in the model. For any $j=1,\dotsc,\ell(\alpha)$ we denote by $I_j(\alpha)$ the $j$-th interval in $\alpha$ and by $r_j(\alpha)$ its length; that is $I_j(\alpha)=(\alpha_{j-1},\alpha_j]$
and $r_j(\alpha) =   \alpha_{j}-\alpha_{j-1}$. We will denote by $r(\alpha)$  the smallest size of such intervals defined by  
\begin{equation}\label{ralpha}
r(\alpha) = \min_{j=1,\dotsc,\ell(\alpha)}\{ r_j(\alpha)\}\,.
\end{equation}

In the sequel we will denote by $\|\cdot\|_r$ the $r$-norm in
$\R^p$. Given tuning parameters $\lambda> 0$, $\gamma>0$ and
  $\delta>0$, and see below for a discussion, we define the joint lasso 
estimator of the change point parameter $\alpha^0$ and the coefficients
$\beta^0(1),\dotsc,\beta^0(\ell(\alpha^0))$ for the
$\ell(\alpha^0)$-segments by 
\begin{align}\label{hatalpha}
\hat\alpha \; &=\; \underset{k}{\arg\min} \;\underset{\alpha\colon\ell(\alpha)=k}{\arg\min}\;\Bigl\{ \;\sum_{j=1}^{k} 
 L_n(I_j(\alpha),\hat\beta(j))
+ \gamma k  \;\Bigr\}\\
\hat\beta(j)\; &=\; \underset{\beta}{\arg\min}\;\Bigl\{\; L_n(I_j(\alpha),\beta) +\lambda \textstyle\sqrt{r_j(\alpha)}\|\beta\|_1   \; \Bigr\}\,,\quad j=1,\dots,\ell(\alpha)\,, \label{hatbeta}
\end{align}
where the loss function $L_n$ is given by
\begin{equation}\label{LXY}
L_n(I_j(\alpha),\beta) = \|\bY_{I_j(\alpha)}-\bX_{I_j(\alpha)}\beta\|_2^2/n
\end{equation}
and the minimization in \eqref{hatalpha} is over the set of all vectors
$\alpha=(\alpha_0,\dotsc,\alpha_k)$ satisfying
$0=\alpha_0<\alpha_1<\dotsc<\alpha_k=1$ and   
$r(\alpha)\geq\delta$. The role of $\delta$ is to ensure that within
  each segment (between two consecutive candidates of change points) there
  are sufficiently many samples ensuring a reasonable accuracy of the
  corresponding estimated regression parameter.
We   sometimes refer to \eqref{hatalpha} as the global estimator which is
  contrasted with a computationally more efficient version in Section~\ref{sec.defBS}.
Note that we do not impose an upper bound for k, but the condition on the
minimal spacing $r(\alpha) \ge \delta$
implies that $k\leq 1/\delta$. 

We propose a cross-validation scheme for ordered data, as
outlined in Section \ref{real-data}, to choose the tuning parameters
$\lambda$ for regularizing with 
respect to high-dimensionality and sparsity and $\gamma$ for regularizing
the number of segments. The ideal value of the parameter $\delta$ is related to
the density of the true underlying change points: 
  in practice, it should be chosen reasonably small such as $\delta = 0.1$
  while from a theoretical view point, one needs $O(\sqrt{\log(p)/n}) \le
  \delta < r(\alpha^0) - O(\sqrt{\log(p)/n})$, 
  i.e., smaller than the minimal distance between the true change points,
  but it cannot be chosen too small (not too fast convergence to zero
  asymptotically) 
  for consistent estimation of the change points and the parameters, as
  described in the theoretical results in Section \ref{mainsec}.

We relate the global estimator by considering the Lasso \cite{tibshirani1996}
for the sub-interval $(u,v]$ 
with $un,vn\in\N$, $vn - un \ge 1$ with parameter
$\lambda/\sqrt{\max(v-u,\delta)}$. It is given by  
\begin{equation}\label{lassobeta}
\hat\beta_{(u,v]} \;=\; \underset{\beta}{\arg\min}\;\Bigl\{\; \|\bY_{(u,v]}-\bX_{(u,v]}\beta\|_2^2/(v-u)n+  \frac{\lambda \|\beta\|_1}{\sqrt{\max(v-u,\delta)}}   \; \Bigr\}\,.
\end{equation}
Observe that the estimator $\hat\beta(j)$ in \eqref{hatbeta} equals 
\begin{equation}\label{lassobeta2}
\hat\beta(j) \;=\; \underset{\beta}{\arg\min}\;\Bigl\{\; L_n(I_j(\alpha),\beta)/r_j(\alpha) +
\frac{ \lambda}{\textstyle\sqrt{r_j(\alpha)}} \|\beta\|_1  \; \Bigr\},
\end{equation}
and therefore, as $r(\alpha)\geq \delta$,  $\hat\beta(j)$ is equal to  the
Lasso estimator in \eqref{lassobeta} with $(u,v] = I_j(\alpha)$;  that is
$\hat\beta(j) = \hat\beta_{I_j(\alpha)}$.  
To compute $\hat\beta(j)$ we can use, for example, the \texttt{R}-package \texttt{glmnet} \cite{glmnet}, and for computing the vector
$\hat\alpha$ in \eqref{hatalpha} we can use dynamic programming as
described next.

\subsection{Exact dynamic programming algorithm}

We present first a dynamic programming approach, known for a long time
\cite[cf.]{hawkins1976}, to compute the estimator in \eqref{hatalpha}. It 
computes the optimum in \eqref{hatalpha} and the estimates in
\eqref{hatbeta}, at the computational cost of $O(n^2 \text{Lasso}(n))$
operations where
$\text{Lasso}(n)$ is the cost to compute the Lasso estimator for a sample
of size $n$ (see also \cite{bai-perron2003} and references therein). 

Let $F_{k}(v)$ denote the minimum value of the function in \eqref{hatalpha} when considering only the sample $(\bY_{(0,v]},\bX_{(0,v]})$ 
and vectors $\alpha$ of size  $\ell(\alpha)=k$;  
that is
\[
F_{k}(v)\; =\; \underset{\alpha\colon
  \ell(\alpha)=k}{\min}\;\Biggl\{\;\sum_{j=1}^k
L_n(I_j(v\alpha),\hat\beta_{I_j(v\alpha)}) + \gamma k  \; \Biggr\}\,.
\]
It is easy to see that the optimal $(k+1)$-dimensional vector $\alpha$
corresponding to $F_{k}(1)$ consists of $k-1$ optimal
change points over $(\bY_{(0,\alpha_{k-1}]},\bX_{(0,\alpha_{k-1}]})$ and a
single segment over $(\bY_{(\alpha_{k-1},1]},\bX_{(\alpha_{k-1},1]})$,  
where $\alpha_{k-1}$ is the rightmost change point proportion. 
Moreover, the $k-1$ segments over
$(\bY_{(0,\alpha_{k-1}]},\bX_{(0,\alpha_{k-1}]})$ must minimize the
function \eqref{hatalpha} for the sample
$(\bY_{(0,\alpha_{k-1}]},\bX_{(0,\alpha_{k-1}]})$, leading to $F_{k-1}(\alpha_{k-1})$.  
In this way, the dynamic programming recursion is computed for any $v\in V_n= \{i/n\colon i=1,\dotsc, n\}$ by 
\begin{align*}
F_1(v) &=\begin{cases}
L_n((0,v],\hat\beta_{(0,v]}) + \gamma\,, & \text{ if }v\geq \delta\,;\\
+\infty\,, &  \text{ if }v < \delta\,.
\end{cases}\\
F_{k}(v) &= \min_{u\in V_n, u<v}\{\,F_{k-1}(u)+
           L_n((u,v],\hat\beta_{(u,v]})+\gamma \,\}\,,\quad k=2,\dotsc, \text{kmax},
\end{align*}
where kmax is an upper bound on $k$ (in our case $\text{kmax}=1/\delta$). 
The estimator $\hat\alpha$ in \eqref{hatalpha} is computed by tabulating 
$F_1(v)$ for all $v\in V_n$ 
and then by computing $F_2(v)$ for all $v\in  V_n$
and so on up to $F_{\text{kmax}}(v)$.
The optimal value of $k$ is obtained by
the equation
\begin{equation}\label{optk}
\hat k = \underset{k=1,\dots,\text{kmax}}{\arg\min}\{\, F_k(1)\,\} 
\end{equation}
and the vector $\hat\alpha=(0,\hat \alpha_1,\dotsc, \hat\alpha_{\hat k-1},1)$ is given by
\begin{align*}
\hat\alpha_{j-1} &= \underset{u\in V_n, u<\hat\alpha_j}{\arg\min} \{\,F_{j-1}(u)+ L_n((u,\hat\alpha_j],\hat\beta_{(u,\hat\alpha_j]})+\gamma \,\}
\,, \quad j = 2,\dotsc,\hat k\,.
\end{align*}

\subsection{Binary Segmentation algorithm}\label{sec.defBS}

Here we describe an efficient Binary Segmentation algorithm
\cite[cf.]{venka1992,fryzlewicz2014} to approximate the estimator given by 
\eqref{hatalpha}, with computational cost of the
order $O(n \log(n) \text{Lasso}(n))$, where $\text{Lasso}(n)$ is the cost to compute the Lasso estimator for a sample of size $n$. The
algorithm will not compute the global estimator defined in
\eqref{hatalpha}, but we
will nevertheless provide in Section \ref{sec.BS} theoretical guarantees for
the algorithm which are  the same
  as for the global
estimator.  

For $u,v\in V_n= \{i/n\colon i=1,\dotsc, n\}$  denote by 
\begin{equation}\label{defH}
H(u,v) = 
\begin{cases}
L_n((u,v],\hat\beta_{(u,v]}) + \gamma\,, & \text{ if } (v-u)n \geq 1\,;\\
0\,,  & \text{ otherwise}
\end{cases}
\end{equation}
and define
\begin{equation}\label{h}
h(u,v)=  \underset{s\in \{u\}\cup [u+\delta,v-\delta]}{\arg\min} \{\, H(u,s) + H(s,v)\,\}\,.
\end{equation}
The idea of the Binary
Segmentation algorithm is to compute the best single change point for the
interval $(0,1]$ (given by $h(0,1)\neq 0$) and then to
iterate this criterion on both segments separated by this point, until no
more change points are found (due to the penalty in the objective
function). We can describe this algorithm by using a binary tree structure
$T$ with nodes labeled by sub-intervals of the form $(u,v] \in V_n^2$ such that $(v-u)n\geq 1$. The steps of the algorithm are then given by:

\begin{enumerate}
\item Initialize $T$ to the tree with a single root node labeled by $(0,1]$.
\item For each terminal node $(u,v] $ in $T$ compute $s = h(u,v)$. If $s>u$ add to $T$
the additional nodes $(u,s]$ and $(s,v]$ as
descendants of node $(u,v]$.
\item Repeat 2. until no more nodes can be added to $T$.
\end{enumerate} 

The set of terminal nodes in $T$,  denoted by $T^0$,  will produce the estimated change point vector $\hat\alpha^{bs}$, by picking up the extremes in these intervals; that is
\[ 
\hat\alpha^{bs} = \bigcup_{(u,v]   \in T^0}\{u,v\}\,.
\]

\section{Theoretical properties}\label{mainsec}

In this section we present the main theoretical results for the global
estimator in \eqref{hatalpha}, which can be computed with dynamic
programming, as well as for the binary segmentation algorithm. 
In the sequel we denote by $S(\beta)$ the support of a parameter vector $\beta$,
given by  $S(\beta)=\{i\colon \beta_i\neq 0\}$. Our assumptions are as
follows. 

\subsection*{Assumption 1.} There exists $K_X<\infty$ 
such that
\[
\|X_i\|_\infty\leq K_X
\]
and $\E(X_i)=0$ for all $i$. 

\subsection*{Assumption 2.} There exists $\sigma^2<\infty$ such that 
\[
 \E(\epsilon_i^2) \leq \sigma^2
 \]
and $\E(\epsilon_i) = 0$ for all $i$.

\subsection*{Assumption 3 (compatibility condition \cite{vandeGeer:07a}).}

The covariance matrix $\Sigma$ is positive definite and the compatibility condition holds for  $\Sigma$ and the set $S_*=\cup_{j=1}^{k_0} S(\beta^0(j))$, with constant $\phi_*>0$. That is, for all  $\beta\in \R^p$ that satisfy $\|\beta_{S_*^c}\|_1\leq 3\|\beta_{S_*}\|_1$ it holds that
\begin{equation}\label{compatibility}
\|\beta_{S_*}\|_1^2 \;\leq\; \frac{\bigl(\beta^T\Sigma \beta \bigr)s_*}{\phi_*^2}
\,,
\end{equation}
where $s_*$  is the cardinality of $S_*$, see also \cite[Ch.6.2.2]{peter-sara-book}. 

\medskip
We note that the compatibility constant $\phi_*^2$ is always lower-bounded by
the minimal eigenvalue of $\Sigma$.  

For any $0\leq i \leq j < k \leq k_0$ 
 denote by
\[
\gamma(i,j,k)=  \frac{\alpha^0_{j}-\alpha^0_{j-1}}{\alpha^0_k-\alpha^0_{i-1}}\,.
\]

We assume the following condition on the vectors $\beta^0(1),\dotsc,\beta^0(k_0)$ to guarantee the identifiability of the model parameters. 

\subsection*{Assumption 4 (identifiability).} If $k_0>1$  there exists a constant 
$m_*>0$ such that 
\[
\min_{1\leq i \leq j < k\leq k_0}   \frac{\|  \sum_{r=i}^j \gamma(i,r,j) \beta^0(r) -  \sum_{r=j+1}^k \gamma(j+1,r,k) \beta^0(r) \|_1}{s_*} \;\geq \; m_*  \,.
\]

\medskip
\begin{remark}\label{identbetas}
Observe that in the case $k_0=2$ the first condition amounts to say that 
$\|\beta^0(1)-\beta^0(2)\|_1\geq m_* s_*$. 
\end{remark}

\medskip
We will denote by $M_*$ the minimal upper bound such that 
\[
\max_{1\leq j \leq k_0}  \|\beta^0(j)\|_\infty \;\leq\; M_*, \quad\text{and
  if $k_0 > 1$ also:}\ \max_{1< j \leq k_0}  \|\beta^0(j-1) - \beta^0(j)\|_\infty \; \leq \; M_*\,.  
\]
Given $K_X$, $\phi_*$, $M_*$ and $m_*$ specified by Assumptions~1-4 we define the constants
\begin{equation*}
d_* =  \begin{cases}
\frac{m_*^2\phi_*^2}{32M_*} & \text{ if }k_0>1\\
+\infty & \text{ if }k_0=1,
\end{cases} 
\end{equation*}
and 
\begin{equation*}
c_*  = \Bigl(\frac{K_XM_*}{d_*}+\frac{\sqrt{8}}{\phi_*} \Bigr)^2\,.
\end{equation*}

\subsection{Global estimator with dynamic programming}

For the global estimator in \eqref{hatalpha} computed by dynamic
programming, we present here a finite-sample result. The corresponding
constants are not of main interest, and an asymptotic interpretation
presented afterwards leads to simpler statements. 
\begin{theorem}\label{modsel} 
Suppose Assumptions 1-4 hold. Given $t>0$, let $\lambda$, $\delta$, $s_*$
and $\gamma$ satisfy: 
\begin{enumerate} 
\item $\delta + \lambda\sqrt\delta/d_*<  r(\alpha^0)$, \\[-10pt]
\item $\lambda/\sqrt\delta < M_*\phi_*^2/24$ and $\lambda\sqrt{\delta} \geq \lambda_0$, with $\lambda_0=40 t \sigma K_X\sqrt{ \frac{\log(n p)}{n}}$, \\[-10pt]
\item $s_* <  \frac{\lambda_1^{-1}}{4c_*}$ 
 with $\lambda_1=  10tK_X^2\sqrt{ \frac{\log(np)}{n}}$,
  \item $\gamma>6c_*\lambda^2 s_*$ and $\gamma +2\lambda\sqrt\delta M_*s_*
    < 4d_*M_* s_*\delta$.
\end{enumerate}
Then, with probability at least $1-  2/t^2$ we have that
\begin{enumerate}
\item$ \ell(\hat\alpha) = k_0$,\\[-10pt]
\item $\|\hat \alpha  -\alpha^0 \|_1\,\leq\,\lambda\sqrt\delta/d_*$ \\[-10pt]
\item $\sum_{j=1}^{k_0} \left(\|\bX_{I_j(\hat\alpha)}(\hat\beta^{(j)} - \beta^0(j))
  \|_2^2/n \,+\lambda \sqrt{r_j(\hat\alpha)}\|\hat\beta^{(j)} -\beta^0(j)\|_1\right)\;
  \leq \; 4c_* k_0\lambda^2 s_*$.
  \end{enumerate}
\end{theorem}

\noindent
\emph{Asymptotic interpretation.} For simplifying the discussion, assume
that $p \gg n$, $K_X = O(1)$, $\sigma = O(1)$ and 
that $\phi_*^2,M_*, c_*, d_*$ are all behaving like $\asymp O(1)$ (bounded
away from zero and bounded above by a fixed constant). We
then distinguish two cases, namely where $r(\alpha^0) \asymp \delta \asymp
O(1)$ and where $r(\alpha^0) \asymp \delta = o(1)$. 

For $r(\alpha^0) \asymp O(1)$ (bounded away from zero), saying that the change
points are well separated 
and there are only finitely many of them, we obtain $\lambda \asymp
\sqrt{\log(p)/n}$, $\lambda_1 \asymp \sqrt{\log(p)/n}$ and we thus require 
that the sparsity $s_* = O(\sqrt{n/\log(p)})$. This is a rather standard
assumption for establishing an oracle inequality with $\ell_1$-norm control
over the estimated parameter (as in statement (3)), see
\cite[Th.6.2-6.3]{peter-sara-book}. We
then obtain the following convergence rates which are analogous as for the
Lasso in a standard high-dimensional sparse linear model:
\begin{eqnarray*}
& &\|\hat{\alpha} - \alpha^0\|_1 = 
O_P(\sqrt{\log(p)/n}),\\
& &\sum_{j=1}^{k_0} \|\bX_{I_j(\hat\alpha)}(\hat\beta^{(j)} - \beta^0(j))
    \|_2^2/n = O_P(s_* \log(p)/n),\\ 
& &\sum_{j=1}^{k_0} \|\hat{\beta}^{(j)}
  -\beta^0(j)\|_1 = O_P(s_* \sqrt{\log(p)/n}).
\end{eqnarray*}
For $r(\alpha^0) \asymp \delta = o(1)$, the conditions require that
$\delta^{-1} = 
O(\sqrt{n/\log(p)})$, i.e., $\delta$ cannot converge faster to zero than
$\sqrt{\log(p)/n}$. In this regime where the change points can be
$O(\sqrt{\log(p)/n})$-dense and where there can be a growing number
thereof, we obtain the results for ``the minimal within segments sample
size'' $\delta n$. That is, $\lambda \asymp O(\sqrt{\log(p)/(\delta n)})$,
and we 
require again that the sparsity $s_* = O(\sqrt{n/\log(p)})$. The
convergence rates 
become
\begin{eqnarray*}
& &\|\hat{\alpha} - \alpha^0\|_1 = 
O_P(\sqrt{\log(p)/n})\ \ \mbox{(independent of $\delta$)},\\
& &\sum_{j=1}^{k_0}
\|\bX_{I_j(\hat\alpha)}(\hat\beta^{(j)} - \beta^0(j)) 
  \|_2^2/n = O_P(s_* k_0 \log(p)/(\delta n)),\\
& &\sum_{j=1}^{k_0}
  \sqrt{r_j(\hat{\alpha})} \|\hat\beta^{(j)} -\beta^0(j)\|_1 = O_P(s_* k_0
  \sqrt{\log(p)/(\delta n)}).
\end{eqnarray*}
One can further distinguish whether $k_0$
  would grow or not, with maximal growth rate of the order $O(\delta^{-1})
  = O(\sqrt{n/\log(p)})$. A most extreme case happens when all the
  change points are equally dense with $k_0 \asymp 
  O(\delta^{-1})$. For the expression $k_0 \sqrt{\log(p)/(\delta n)}$ to
  converge to zero we need that $\delta^{-1} = o((n/\log(p))^{1/3})$, that
  is a somewhat less dense regime, and the sparsity then needs to be of the
  order $s_* = o(\sqrt{\delta^3n/\log(p)}$ to imply that $s_* k_0 \sqrt{\log(p)/(\delta n)}
  = o(1)$). 
We summarize the asymptotic interpretations in Table \ref{tab1}.
\begin{table}[!htb]
\begin{tabular}{l|l|l|l|l}
regime & $\delta \asymp r(\alpha^0)$ & $\lambda$ & $k_0$ & $s_*$ \\
\hline
non-dense & $> O(1)$ & $O(\sqrt{\frac{\log(p)}{n}})$ & $O(1)$ &
                                                        $o(\sqrt{\frac{n}{\log(p)}})$\\
dense, finite $k_0$ & $\gg O(\sqrt{\frac{\log(p)}{n}})$ &
                                                       $O(\sqrt{\frac{\log(p)}{\delta n}})$ & $ O(1)$ &
                                                                 $o(\sqrt{\frac{\delta n}{\log(p)}})$
  \\
equi-dense & $\gg O((\frac{\log(p)}{n})^{1/3})$ & $O(\sqrt{\frac{\log(p)}{\delta n}})$
                                                             &
                                                                     $o((\frac{n}{\log(p)})^{1/3})$
                                      & $o(\sqrt{\frac{\delta^3 n}{\log(p)}})$
 
\end{tabular}
\caption{Different asymptotic regimes such that $s_* k_0 \lambda = o(1)$
  (which ensures convergence to zero for $\sum_{j=1}^{k_0}
  \sqrt{r_j(\hat{\alpha})} \|\hat\beta^{(j)} -\beta^0(j)\|_1$).}\label{tab1}
\end{table}

\subsection{Binary Segmentation algorithm}\label{sec.BS}

For the binary segmentation algorithm we obtain a similar result as for the
global estimator. 
\begin{theorem}\label{modsel2} 
Suppose Assumptions 1-4 hold. Given $t>0$, let $\lambda$, $\delta$, $s_*$
and $\gamma$ satisfy the conditions (1)-(4) in Theorem~\ref{modsel}. 
Then, with probability at least $1-  2/t^2$ we have that
\begin{enumerate}
\item$ \ell(\hat\alpha^{bs}) = k_0$,\\[-10pt]
\item $\|\hat \alpha^{bs}  -\alpha^0 \|_1\,\leq\,\lambda\sqrt\delta/d_*$ and  \\[-10pt]
\item $\sum_{j=1}^{k_0} \left(\|\bX_{I_j(\hat\alpha^{bs})}(\hat\beta^{(j)} - \beta^0(j))
  \|_2^2/n \,+\lambda \sqrt{r_j(\hat\alpha^{bs})}\|\hat\beta^{(j)}
  -\beta^0(j)\|_1 \right) \;
  \leq \; 4c_* k_0\lambda^2 s_*$.
  \end{enumerate}
\end{theorem}

We note that the conditions and statements in Theorem~\ref{modsel2} for
the binary segmentation algorithm are the same as for the global estimator
in Theorem~\ref{modsel}.



\section{Simulation study}\label{algsec}

We evaluate here the performance of the global change point estimator
computed with the
dynamic programming algorithm (DPA) and of the binary segmentation 
algorithm (BSA). In the simulations we considered a two segments model,
with $\alpha^0=(0,0.5,1)$, $\beta^0(1)=(1,1,0,\dotsc,0)$,
$\beta^0(2)=(0,\dotsc,0,1,1)$, and a three segments model with
$\alpha^0=(0,0.3,0.7,1)$, $\beta^0(1)=(1,1,0,\dotsc,0)$,
$\beta^0(2)=(0,\dotsc,0,1,1)$ and $\beta^0(3)=\beta^0(1)$. For both cases
we use the standard deviation of the error $\sigma=1$ and
$X\sim\Nor(0,\Sigma)$, for different structures of $\Sigma$: 
\begin{enumerate}
\item $\Sigma_{ij}=\1_{\{i=j\}}$ for all $i,j$ (the identity matrix); 
\item  $\Sigma_{ij}=0.8^{|i-j|}$ for all $i,j$ (Toeplitz matrix);
\item  $\Sigma_{ij}=1 - 0.8\cdot \1_{\{i\neq j\}}$ for all $i,j$
  (equi-correlation). 
\end{enumerate}
We consider a range of sample sizes and taking as number of  
covariates $p=2n$. For all the simulation results, we always used the
tuning parameters values $\delta=0.25$, 
  $\lambda=\sqrt{\log(p/(\delta n)}$ and $\gamma  =0.25\lambda$ without
  further fine-tuning (for both algorithms). For the computations we used the
  {\tt R} software and  
the package \emph{glmnet} \cite{glmnet} to fit the parameters in each segment.  

The results of the methods are shown in Figures
\ref{fig:sim1}--\ref{fig:sim3}.  
For each sample size we construct boxplots of the
first change point fraction $\hat\alpha_1$ for 100 replications (when
$\ell(\hat{\alpha})=1$ the first change point was treated as missing
value). We also computed the proportion of $\ell(\hat{\alpha})$ in the 100
replications, to illustrate the performance in estimating the
number of segments. 
\begin{figure}[!htb]
\includegraphics[scale=0.35,trim=1cm 1cm 0 0,clip=true]
{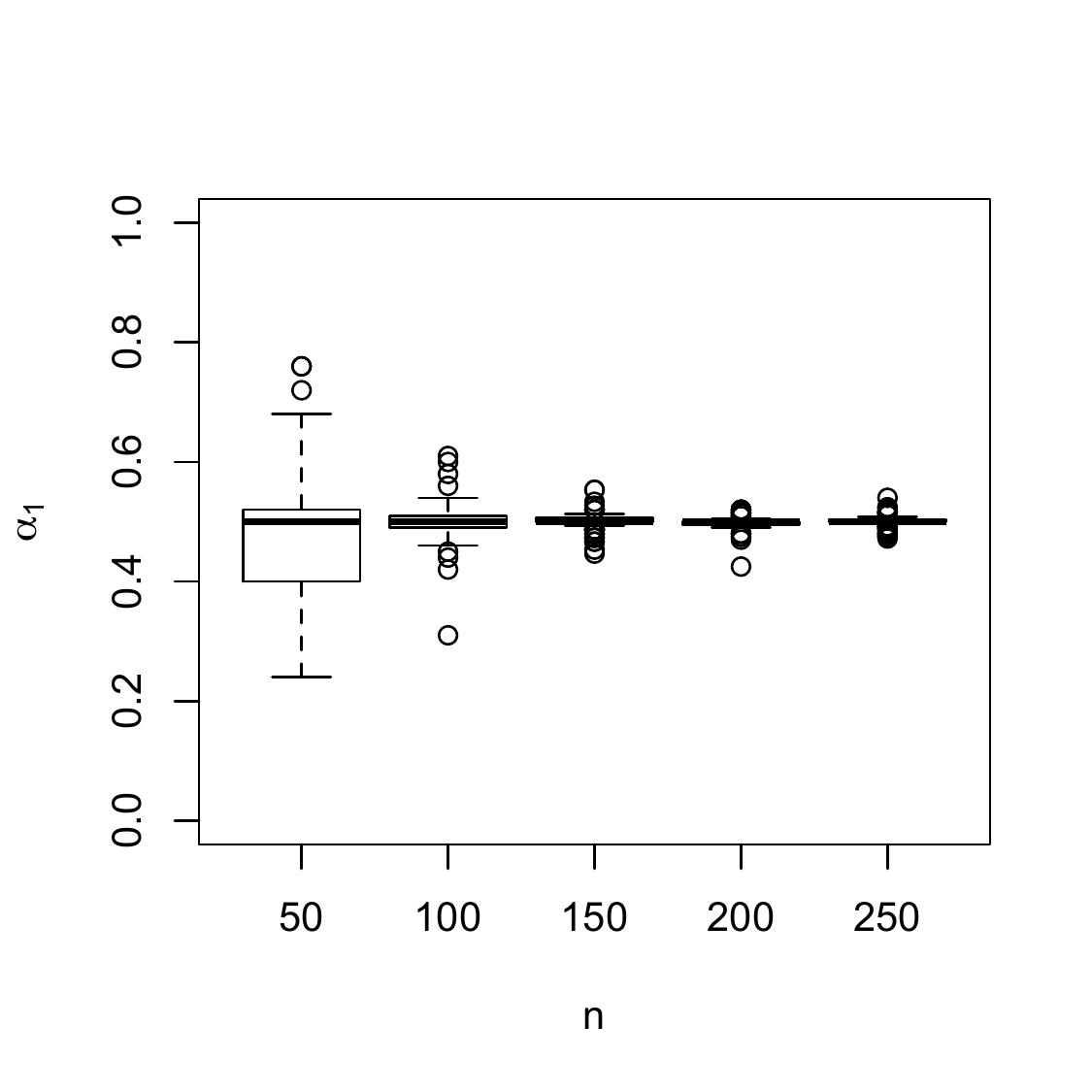}
\includegraphics[scale=0.35,trim=1cm 1cm 0 0,clip=true]
{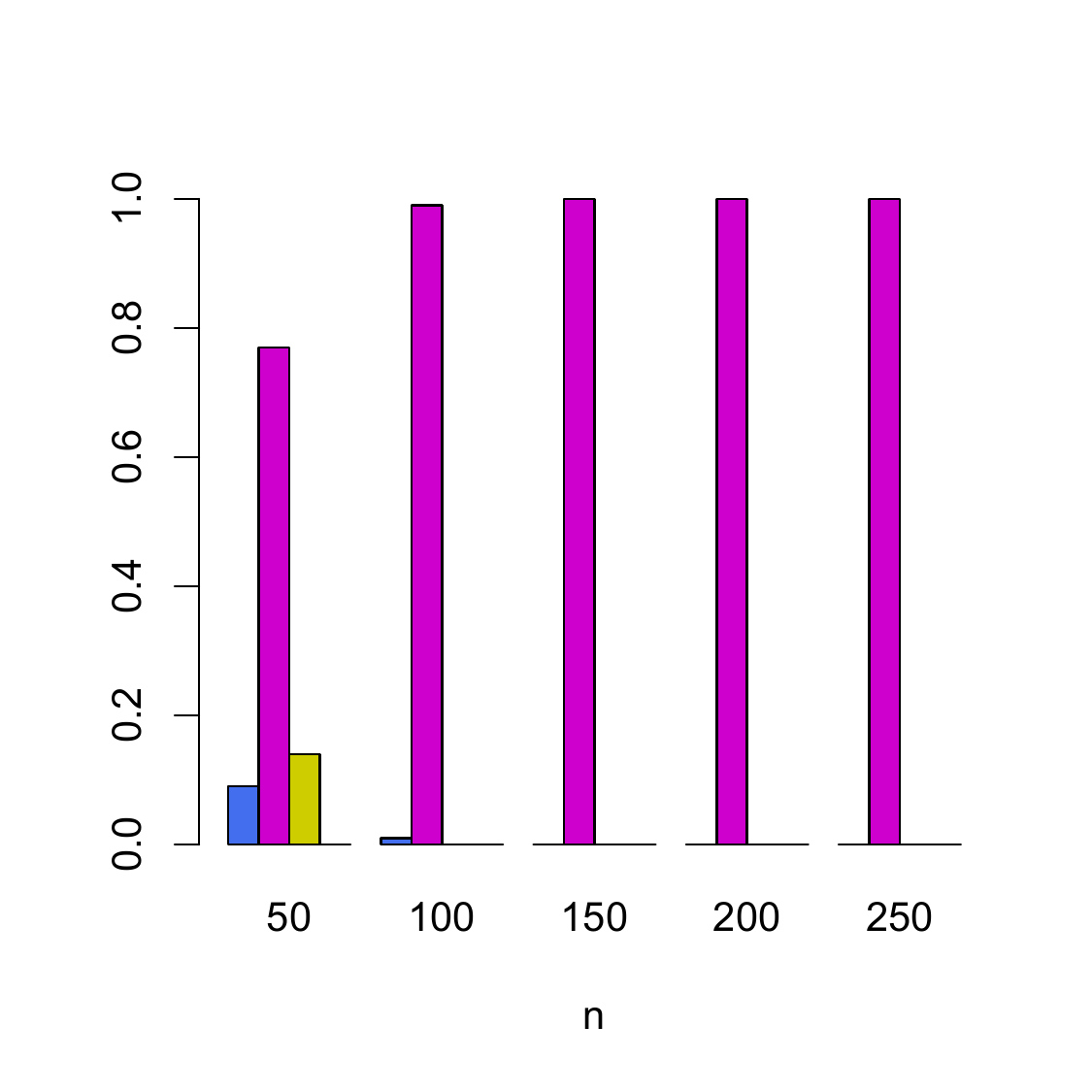}
\includegraphics[scale=0.35,trim=1cm 1cm 0 0,clip=true]
{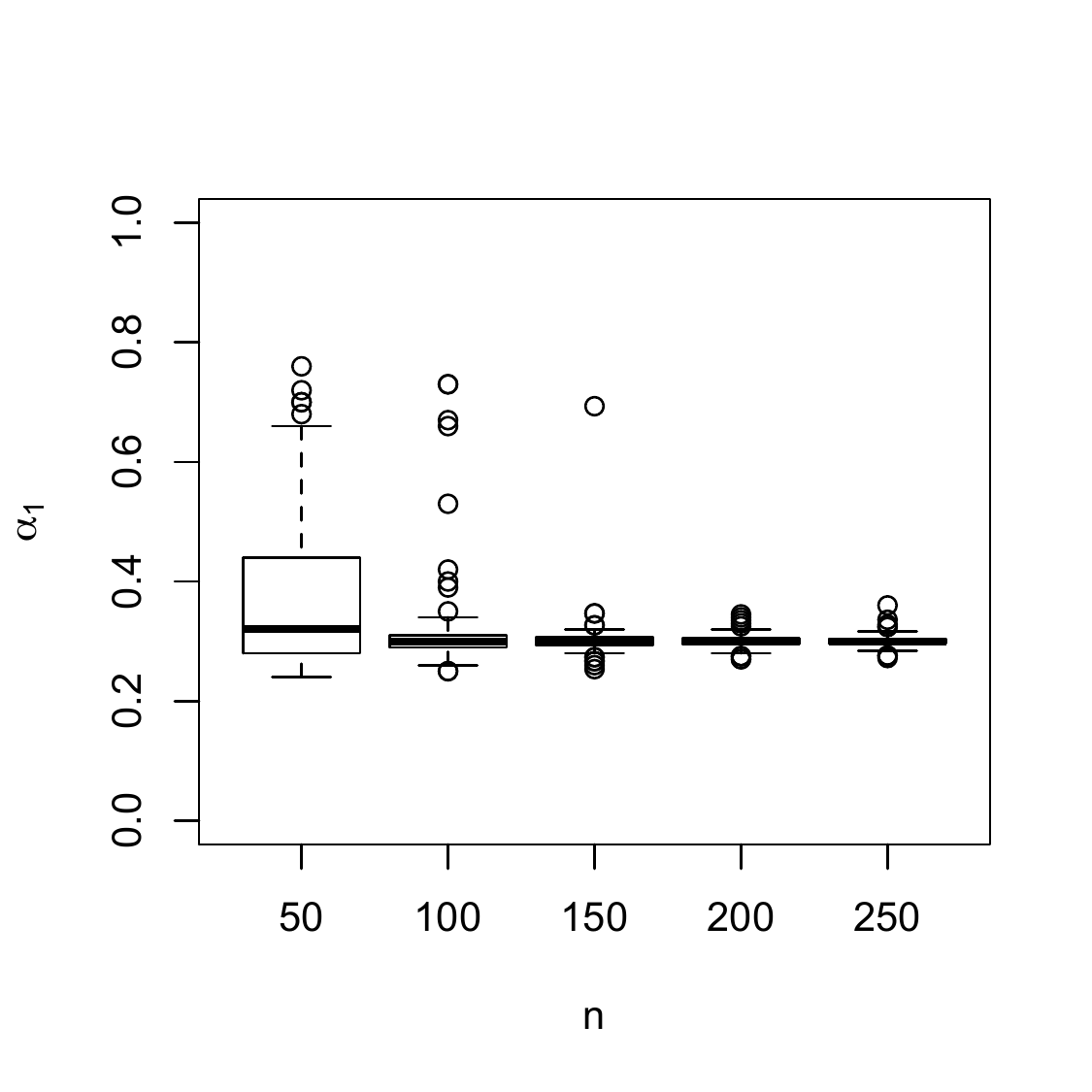}
\includegraphics[scale=0.35,trim=1cm 1cm 0 0,clip=true]
{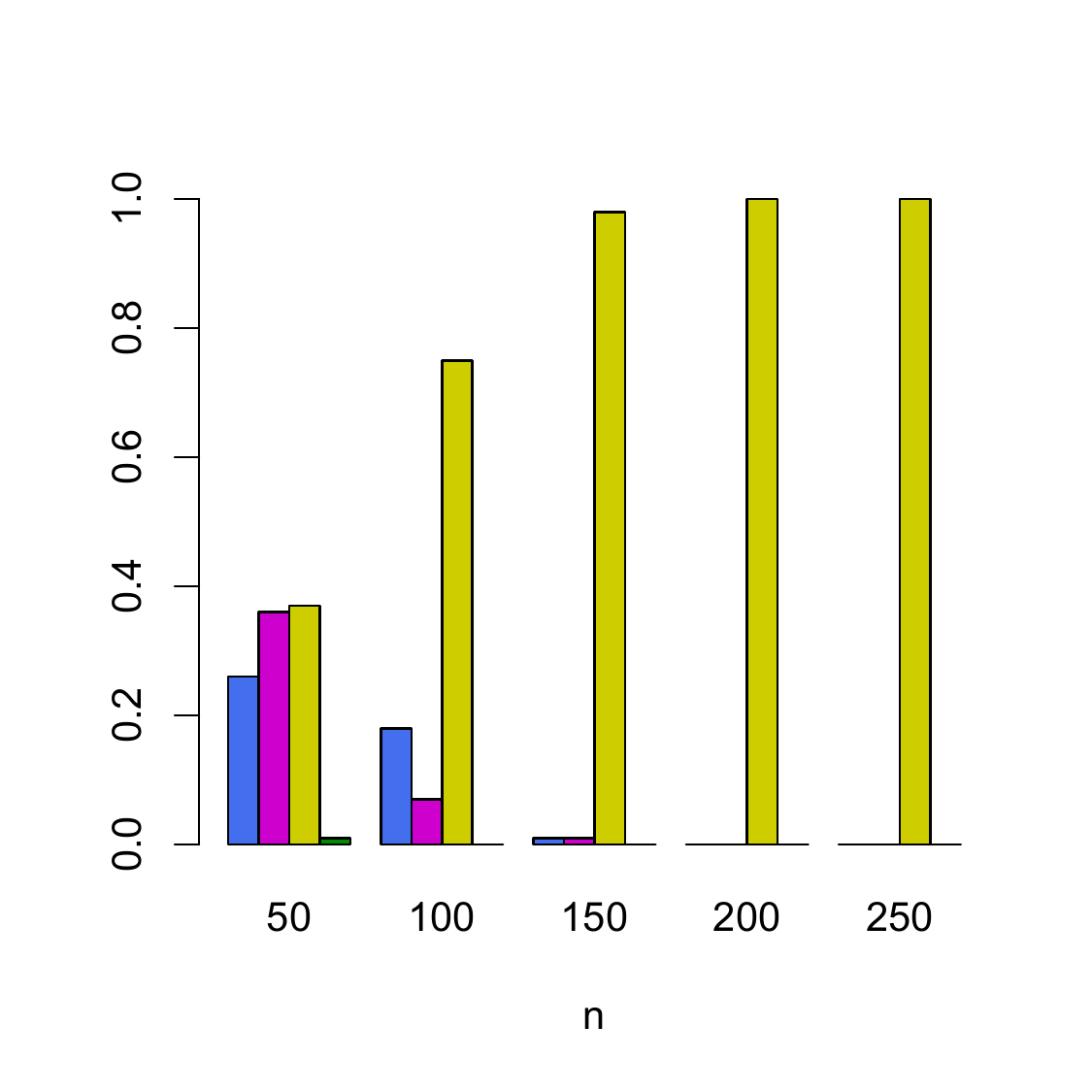}
\\
\includegraphics[scale=0.35,trim=1cm 1cm 0 0,clip=true]
{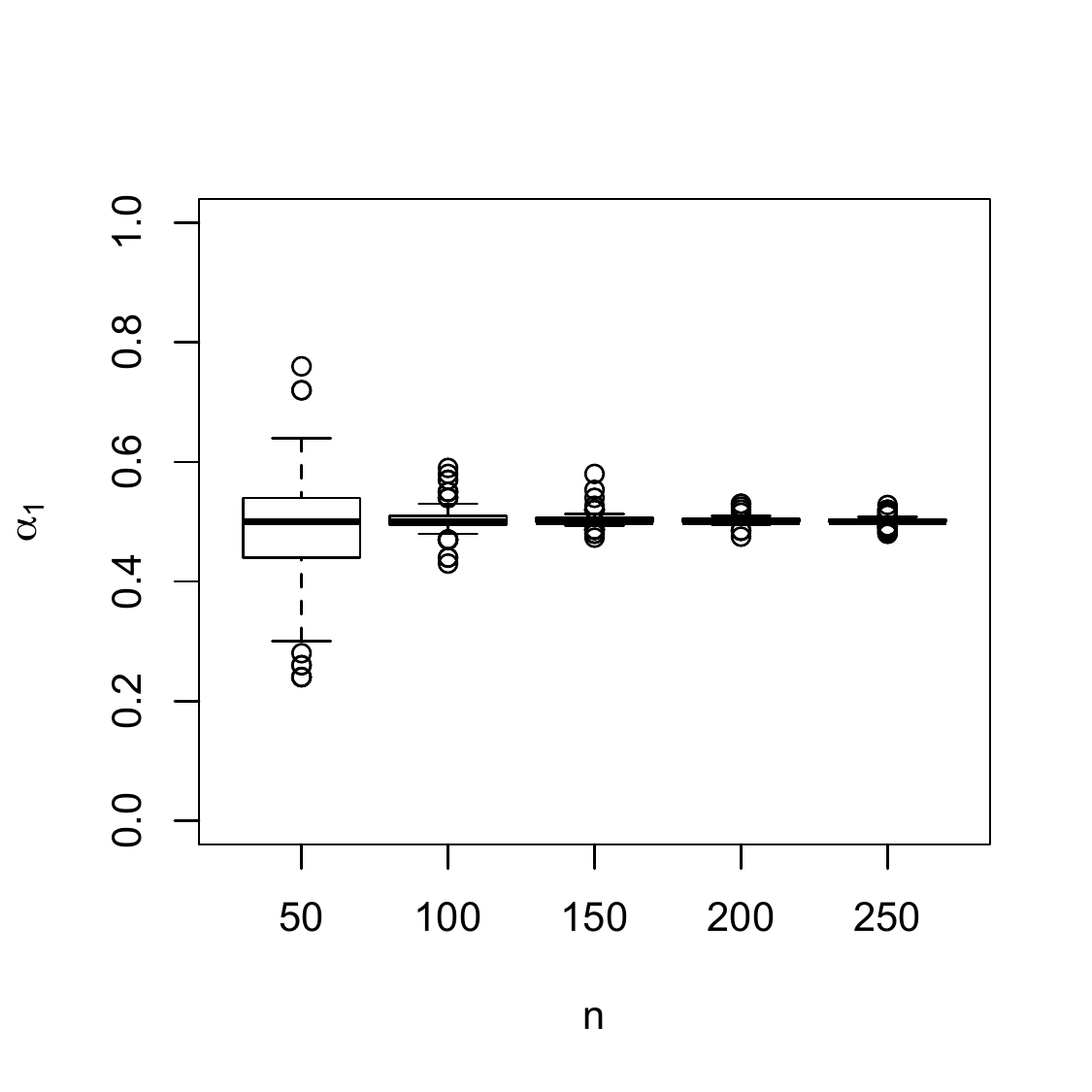}
\includegraphics[scale=0.35,trim=1cm 1cm 0 0,clip=true]
{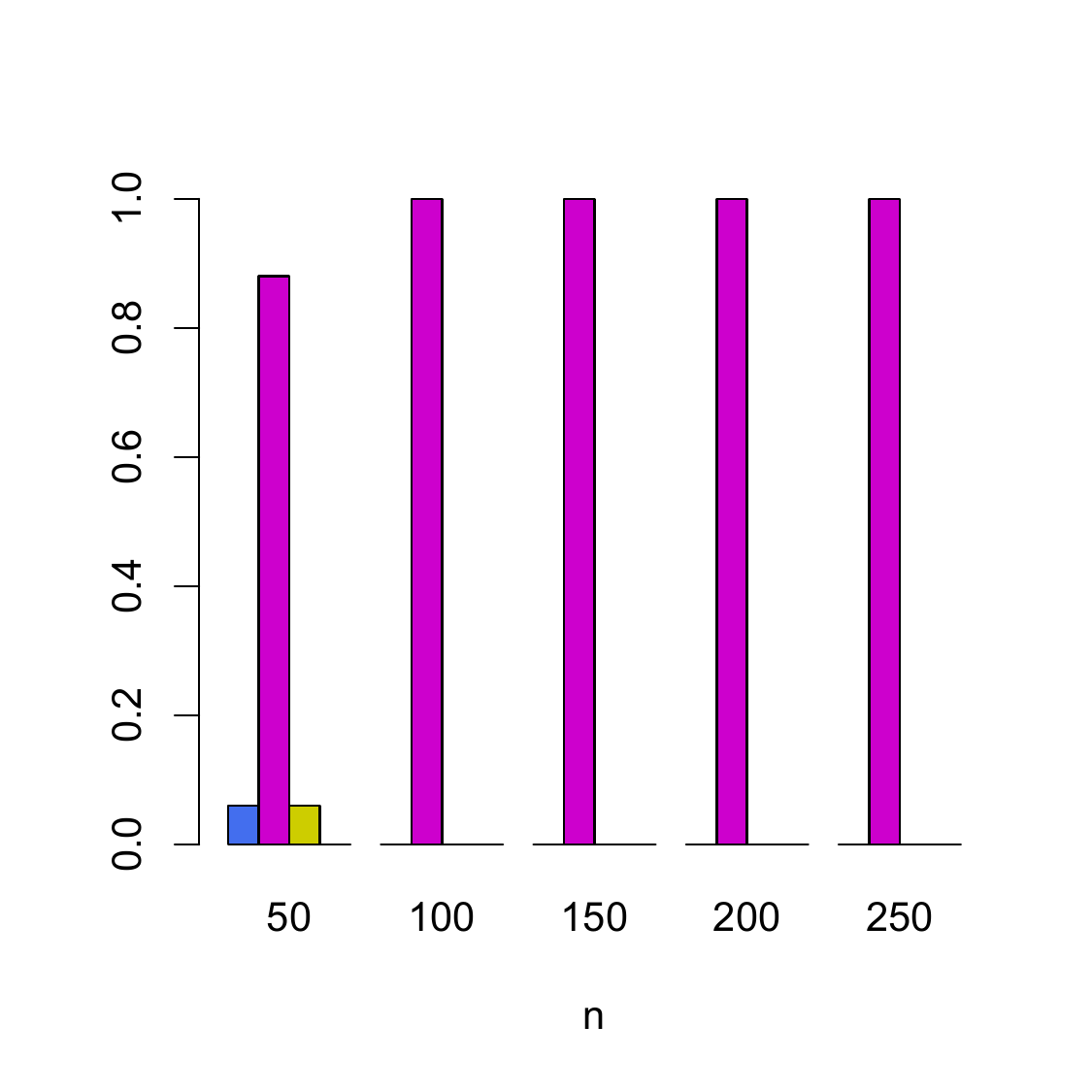}
\includegraphics[scale=0.35,trim=1cm 1cm 0 0,clip=true]
{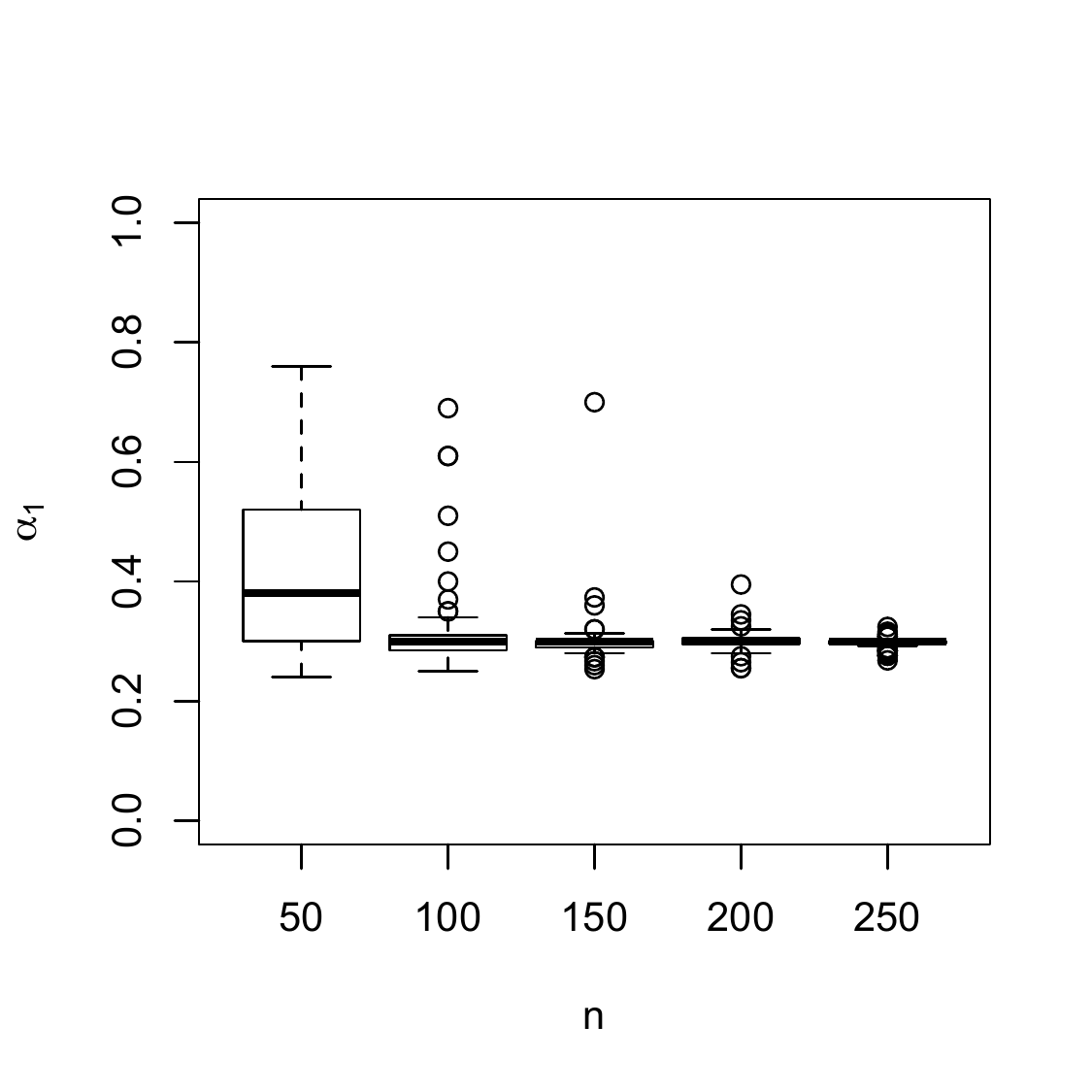}
\includegraphics[scale=0.35,trim=1cm 1cm 0 0,clip=true]
{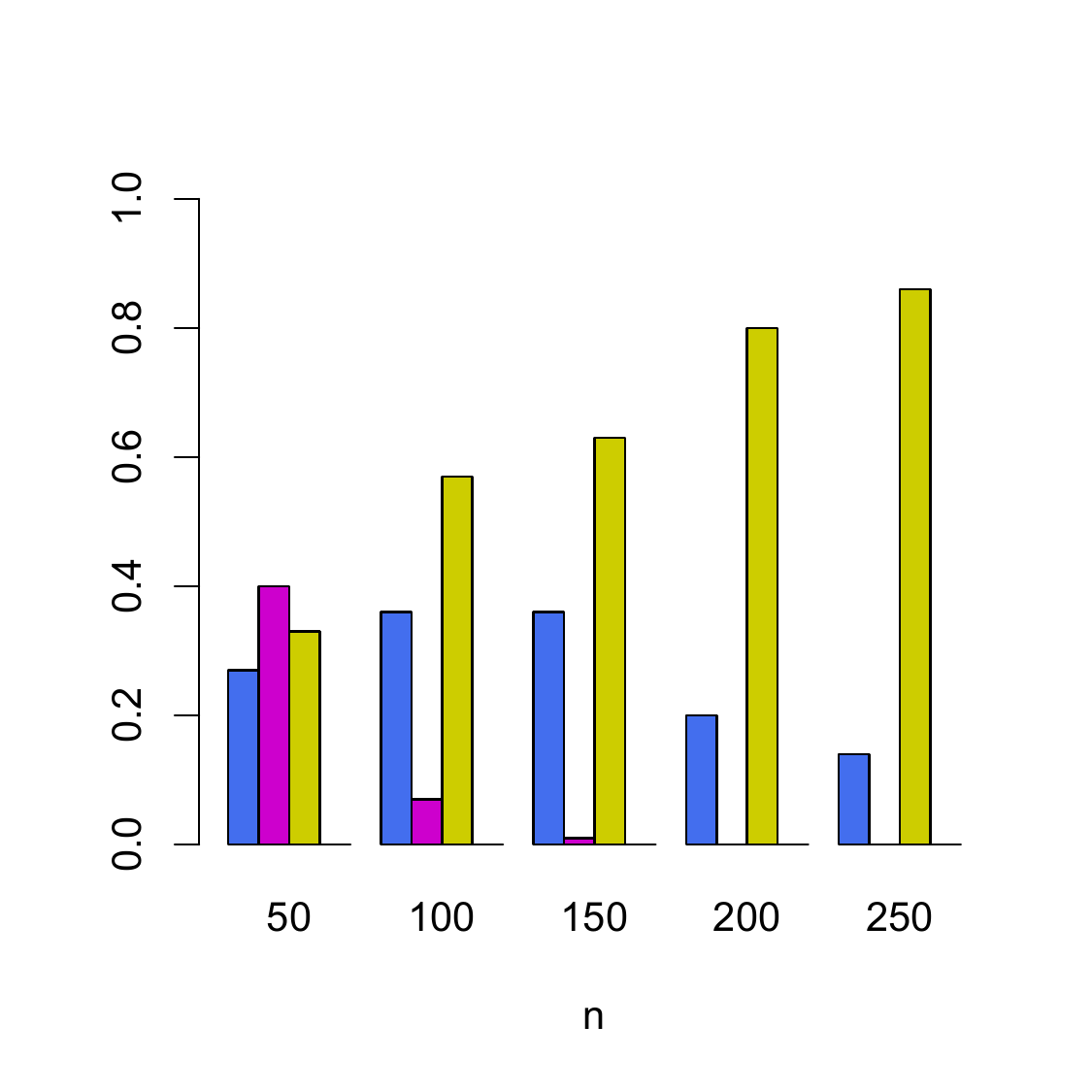}
\caption{First estimated change point fraction $\hat{\alpha}_1$ and number
  of estimated segments $\ell(\hat{\alpha})$, as a function of sample size
  $n$ and $p = 2n$. Model (1) for covariance structure
  $\Sigma_{ij}=\1_{\{i=j\}}$ for all $i,j$. Top: global estimator
  \eqref{hatalpha} using DPA; Bottom: BS-algorithm. Left two panels: two
  segments model with 
  $\alpha^0=(0,0.5,1)$; Right two panels: three segments model with
  $\alpha^0=(0,0.3,0.7,1)$. The barplots correspond to the relative
  frequencies that the algorithm gave a estimated single segment model
  (blue), a two 
  segments model (magenta), a three segments model (yellow) or a four or
  more  segments model (green).} 
\label{fig:sim1}
\end{figure}
\begin{figure}
\includegraphics[scale=0.35,trim=1cm 1cm 0 0,clip=true]
{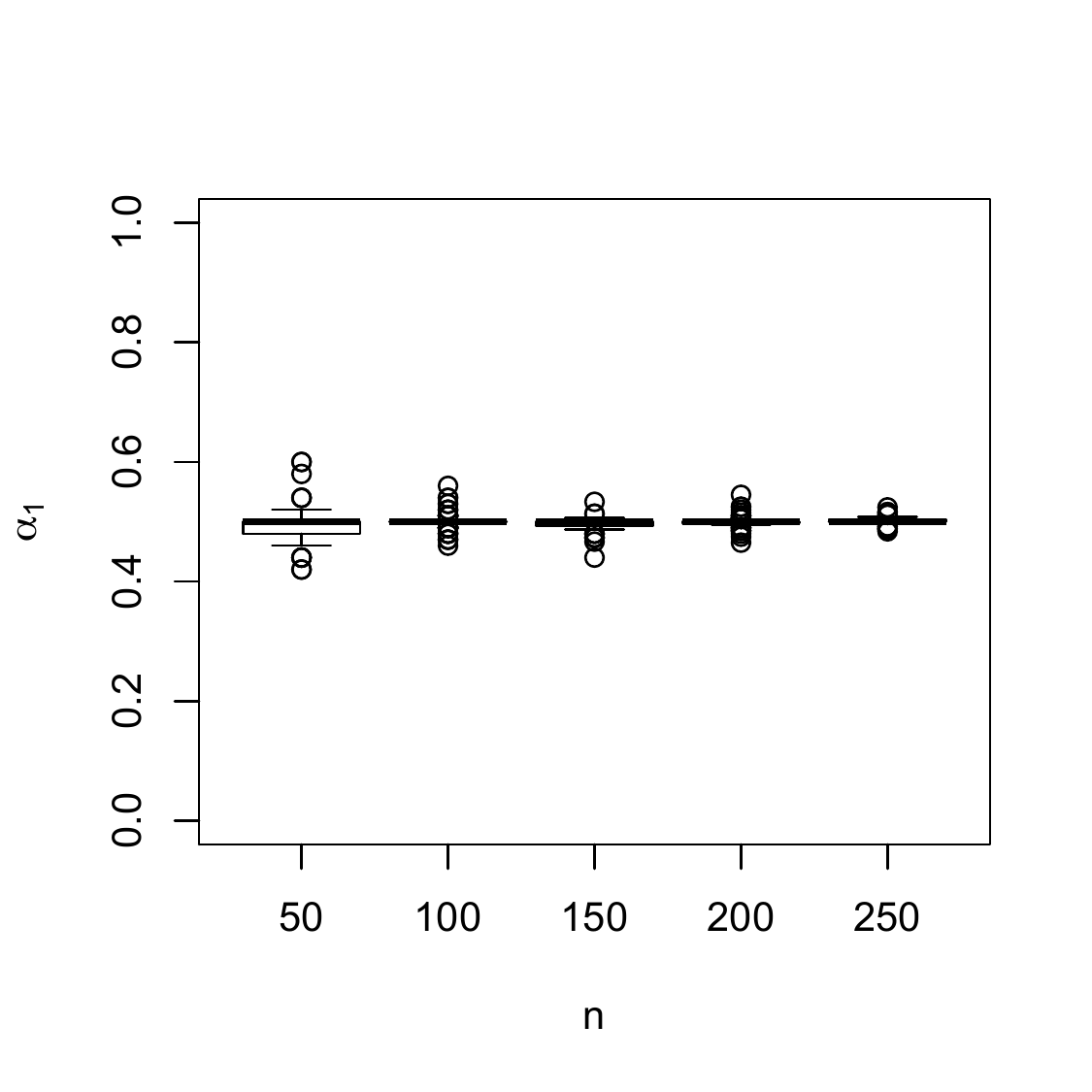}
\includegraphics[scale=0.35,trim=1cm 1cm 0 0,clip=true]
{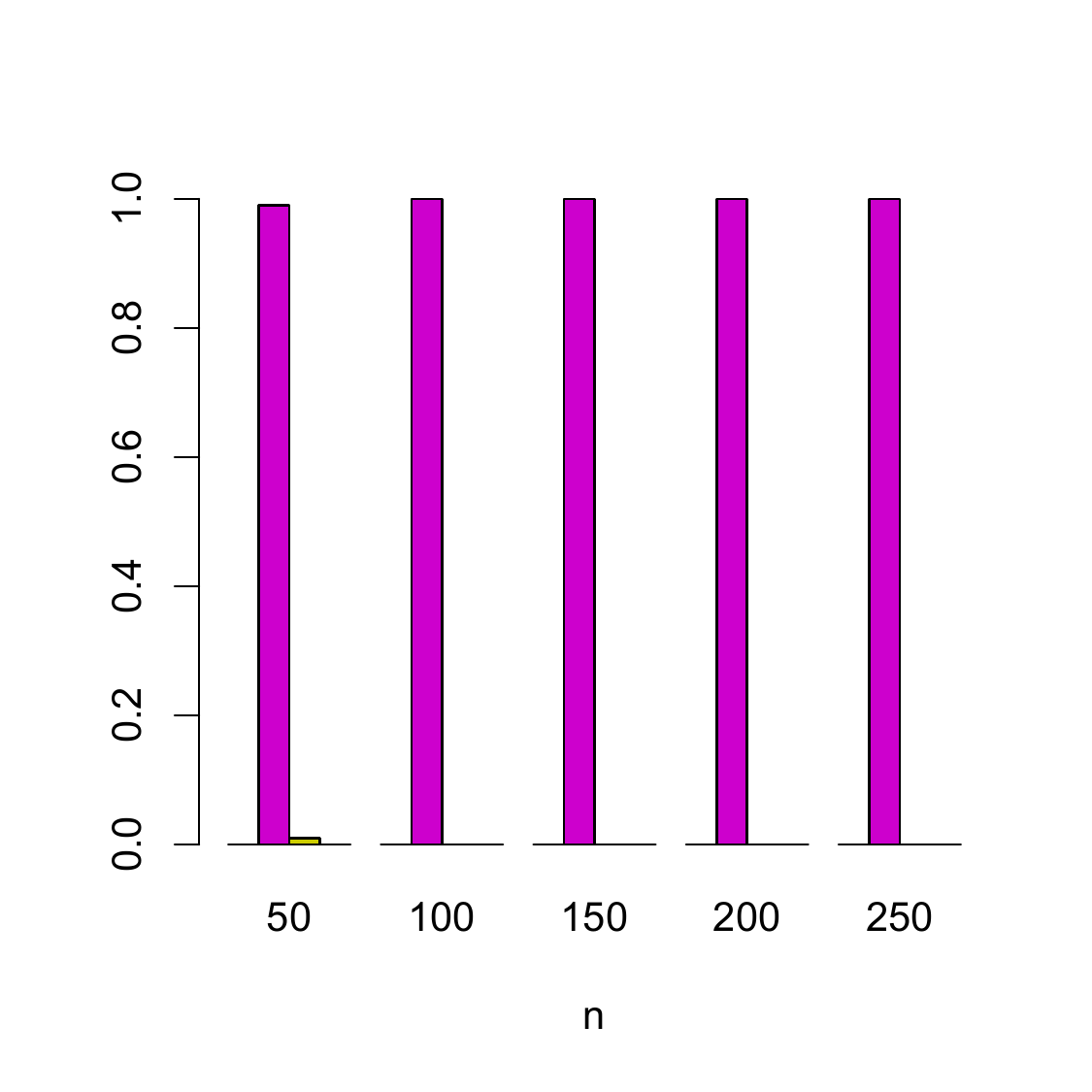}
\includegraphics[scale=0.35,trim=1cm 1cm 0 0,clip=true]
{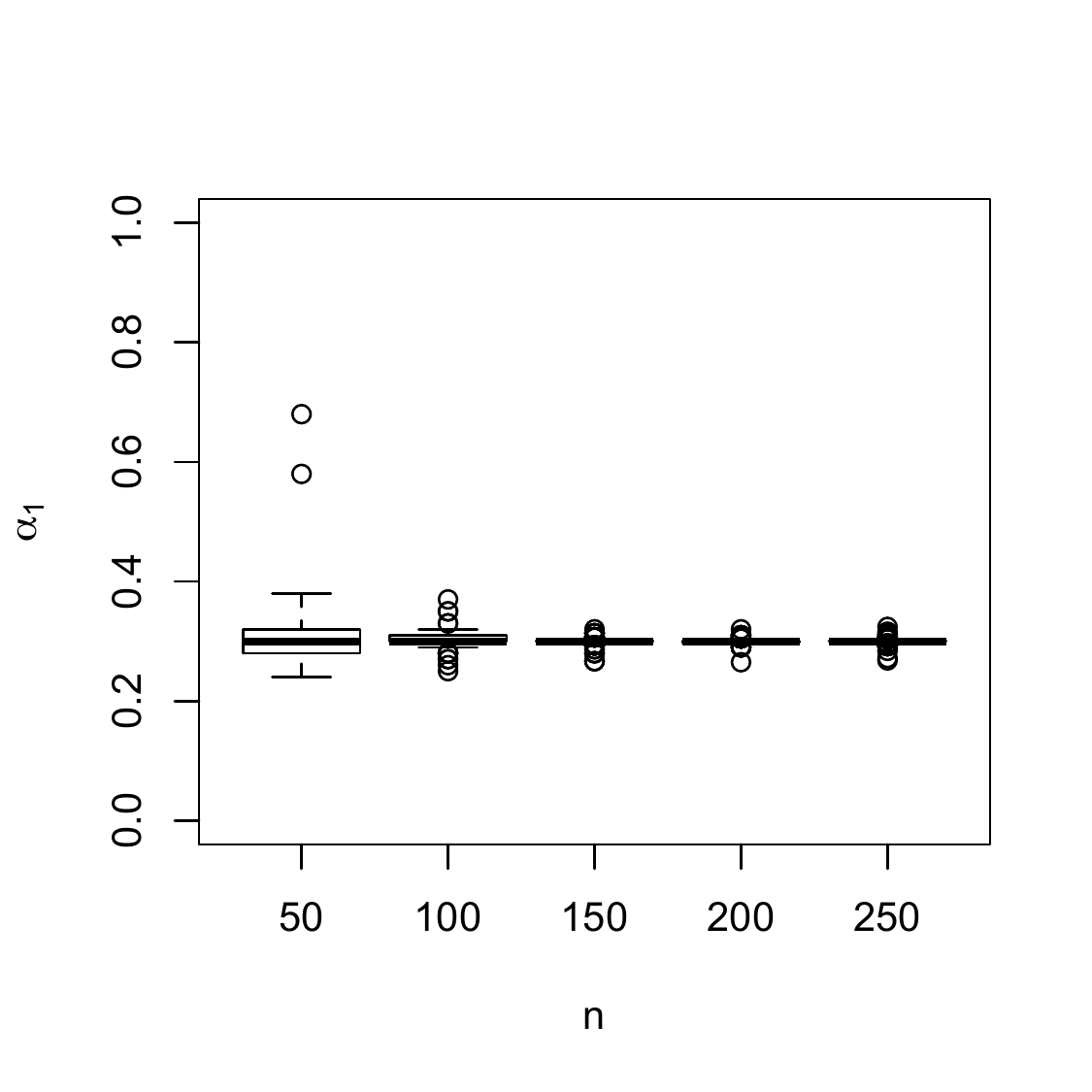}
\includegraphics[scale=0.35,trim=1cm 1cm 0 0,clip=true]
{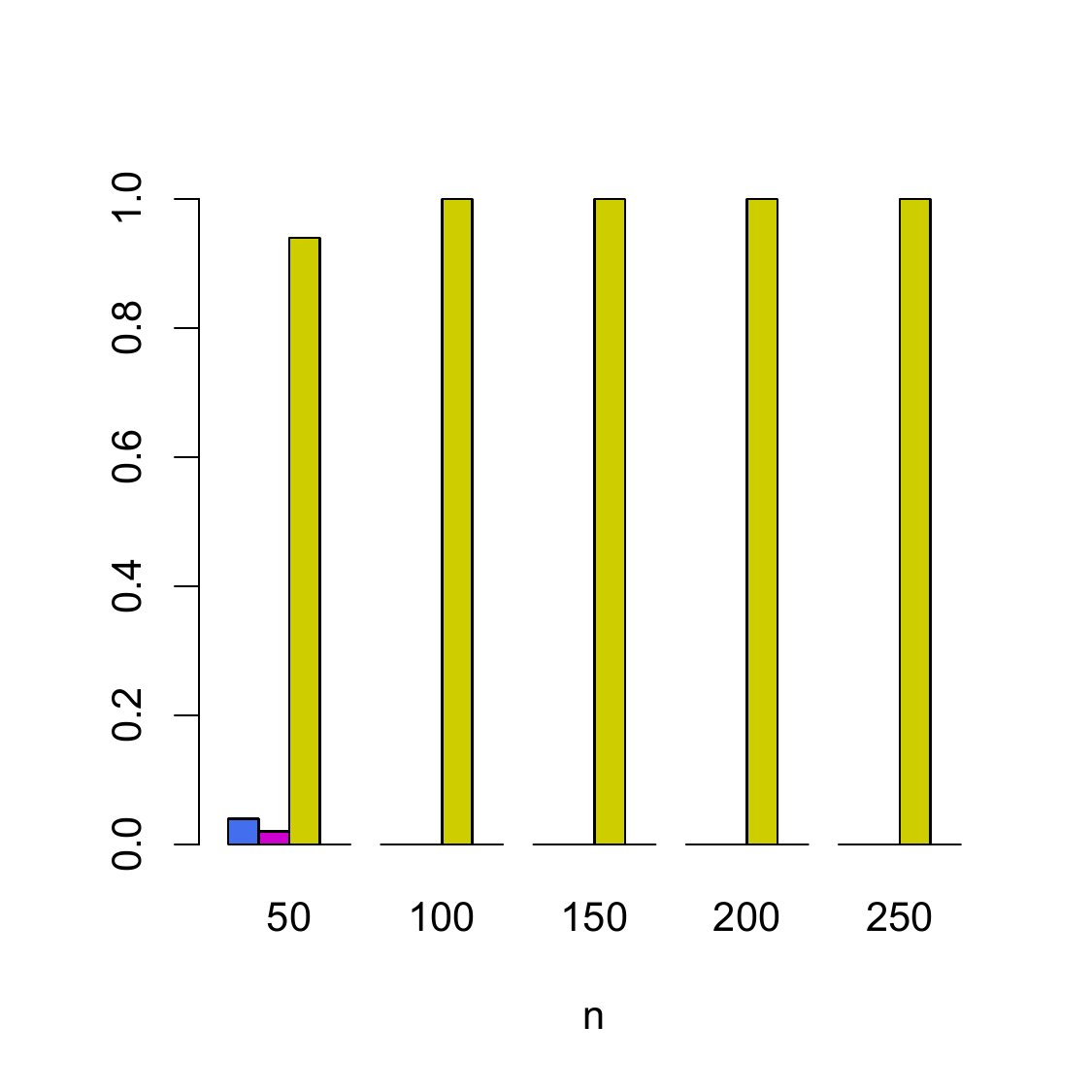}
\\
\includegraphics[scale=0.35,trim=1cm 1cm 0 0,clip=true]
{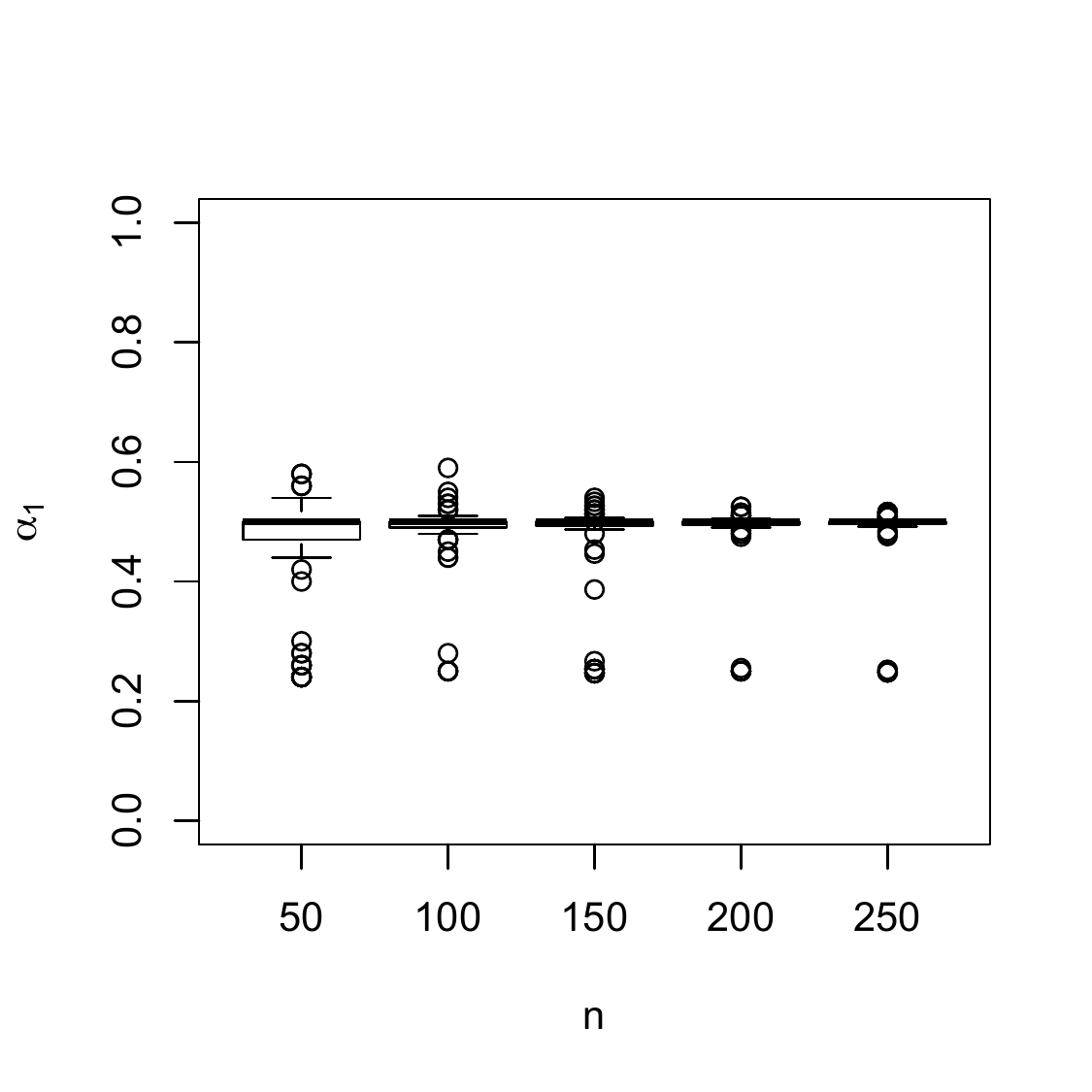}
\includegraphics[scale=0.35,trim=1cm 1cm 0 0,clip=true]
{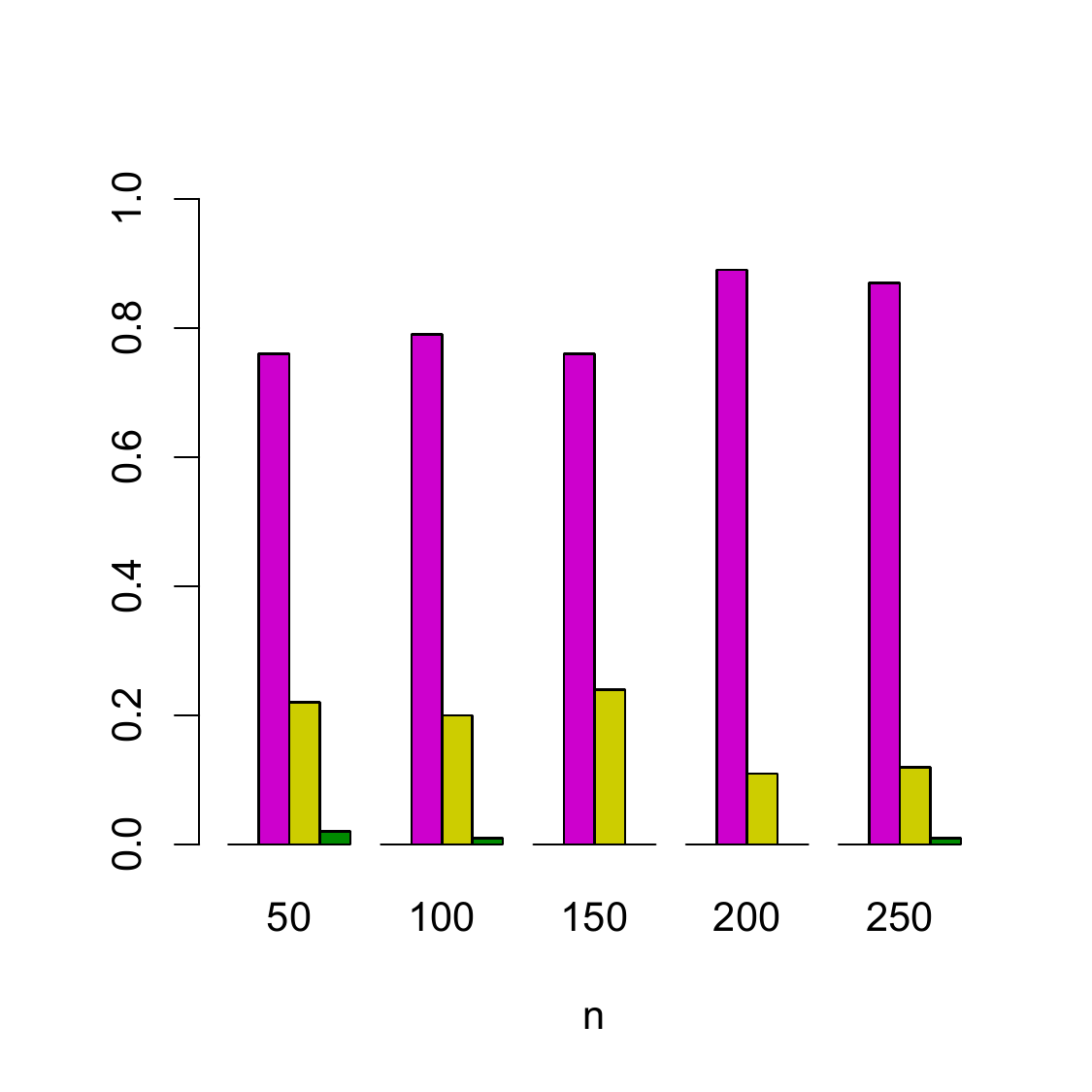}
\includegraphics[scale=0.35,trim=1cm 1cm 0 0,clip=true]
{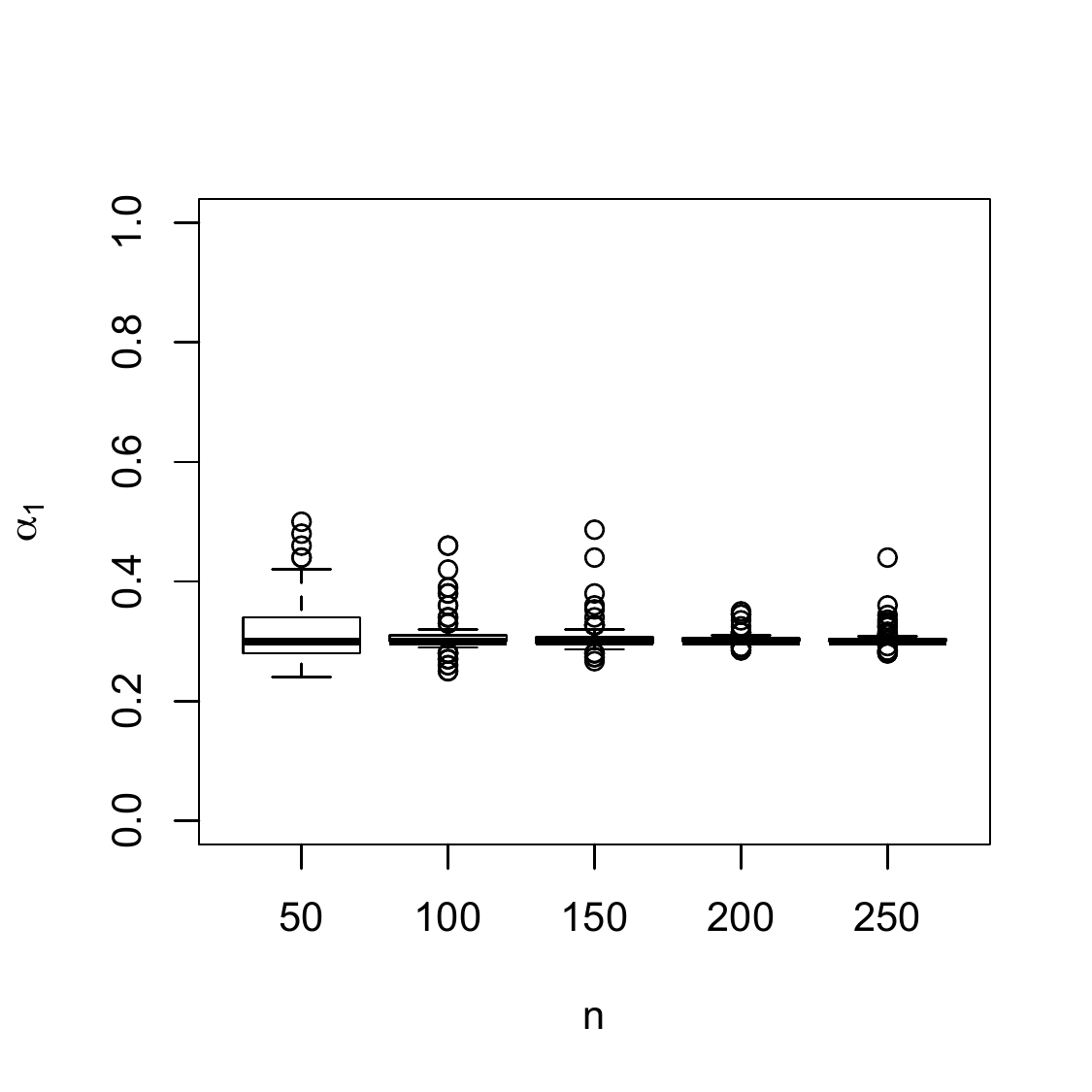}
\includegraphics[scale=0.35,trim=1cm 1cm 0 0,clip=true]
{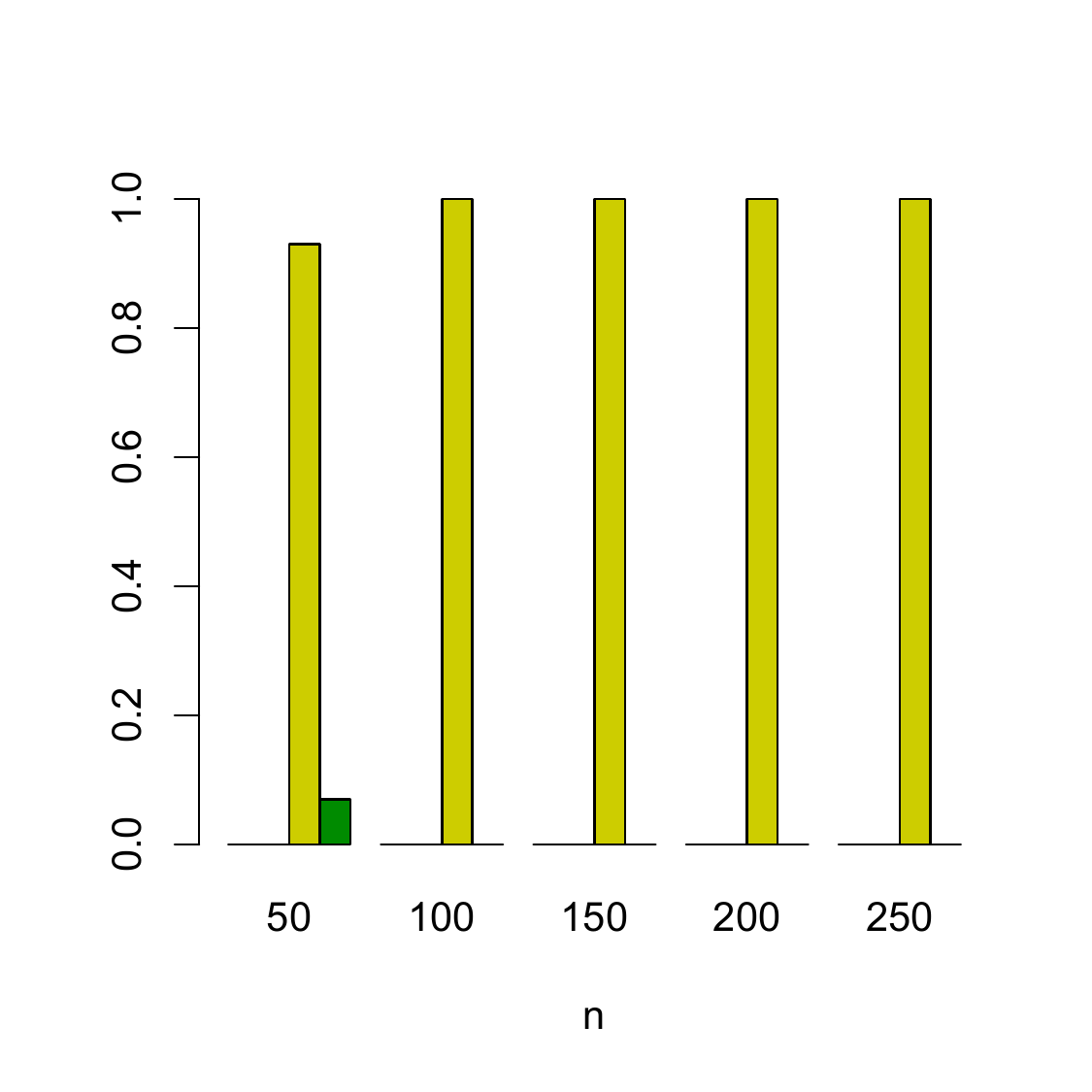}
\caption{As in Figure~\ref{fig:sim1}, but with model (2) for
  covariance matrix $\Sigma_{ij}=0.8^{|i-j|}$ for all $i,j$. Top: global estimator
  \eqref{hatalpha} using DPA; Bottom: BS-algorithm. Left two panels: two
  segments model; Right two panels: three segments model. The
  barplots correspond to the relative
  frequencies that the algorithm gave a
  estimated single segment model (blue), a two segments model (magenta), a
  three segments model (yellow) or a four or more segments model (green).} 
\label{fig:sim2}
\end{figure}
\begin{figure}
\includegraphics[scale=0.35,trim=1cm 1cm 0 0,clip=true]
{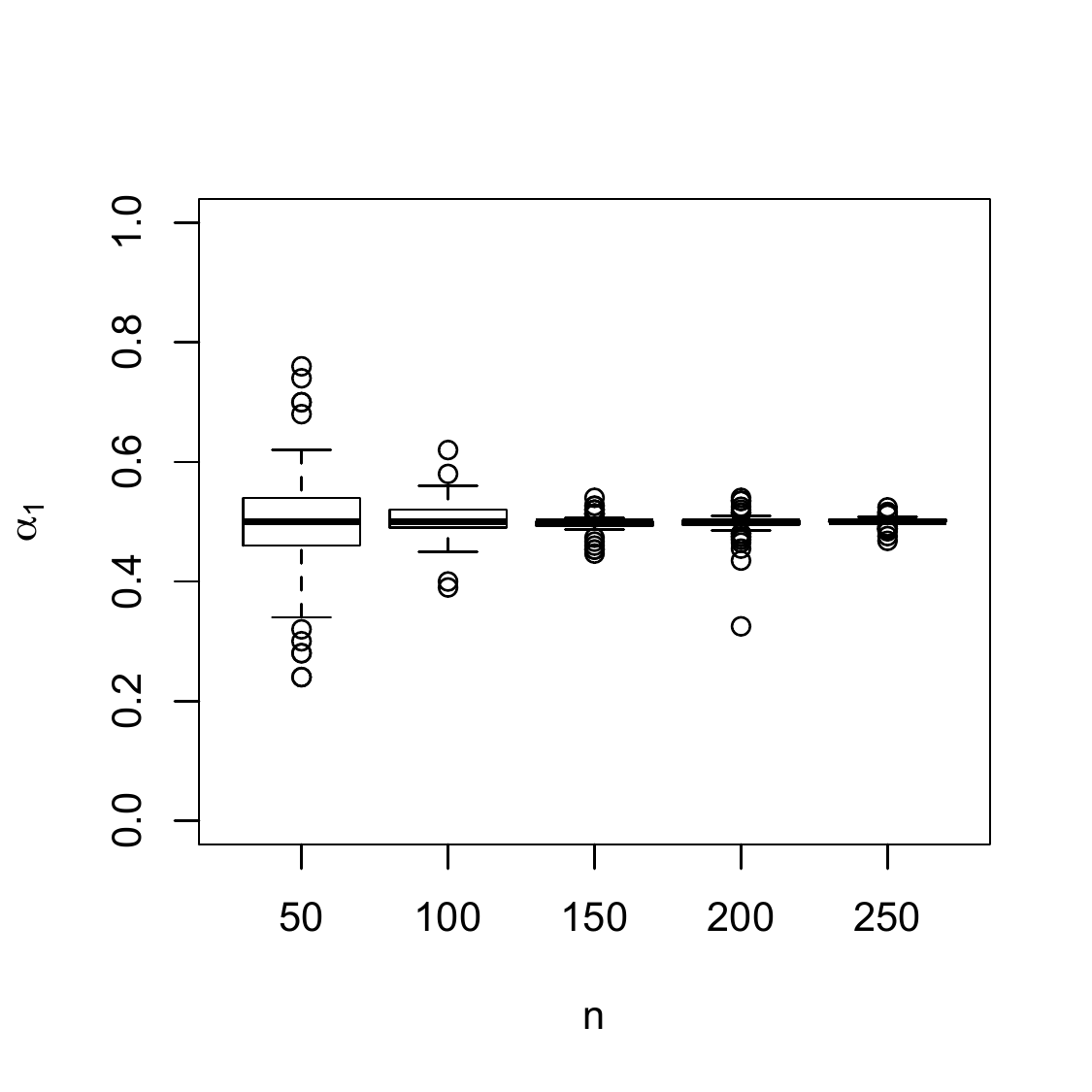}
\includegraphics[scale=0.35,trim=1cm 1cm 0 0,clip=true]
{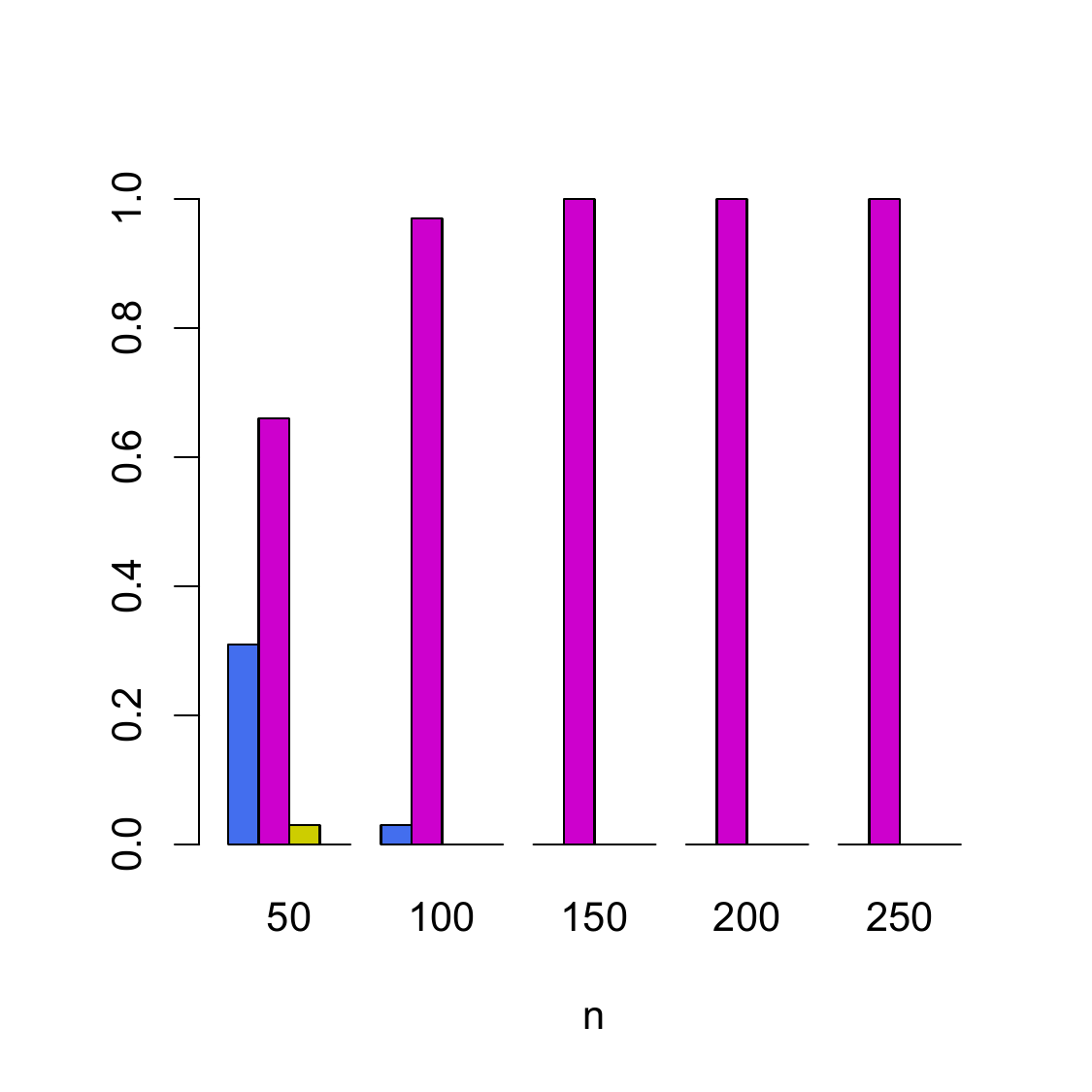}
\includegraphics[scale=0.35,trim=1cm 1cm 0 0,clip=true]
{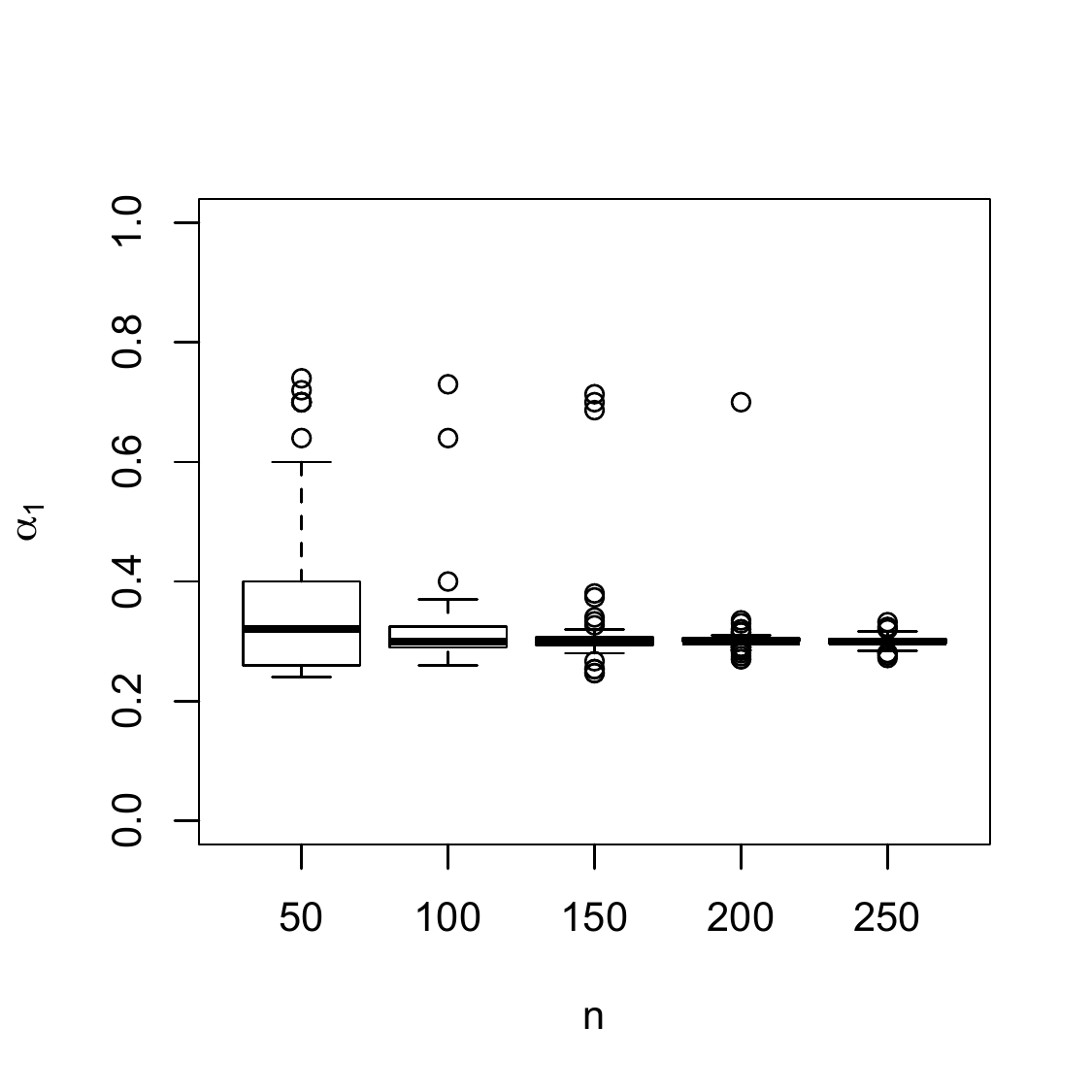}
\includegraphics[scale=0.35,trim=1cm 1cm 0 0,clip=true]
{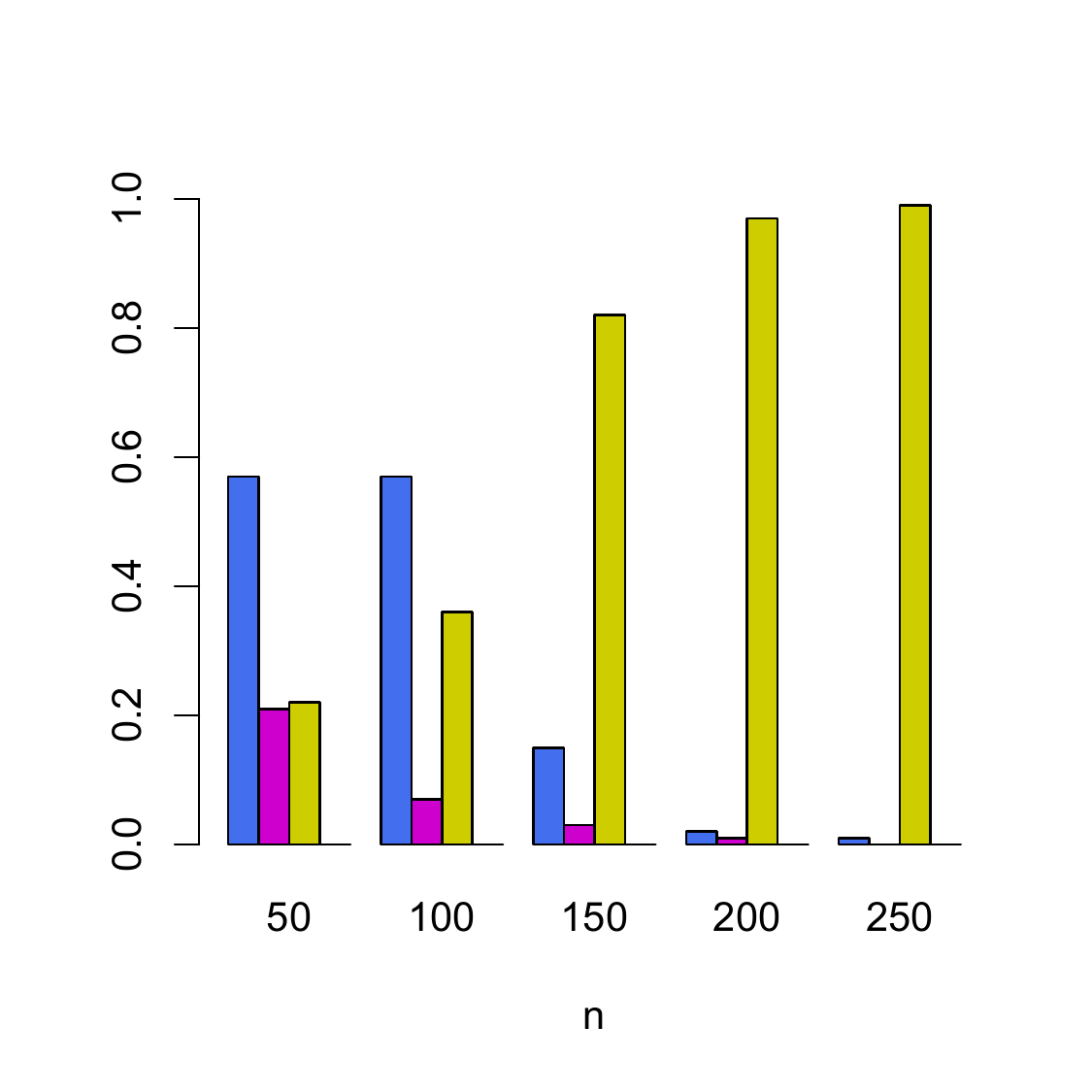}
\\
\includegraphics[scale=0.35,trim=1cm 1cm 0 0,clip=true]
{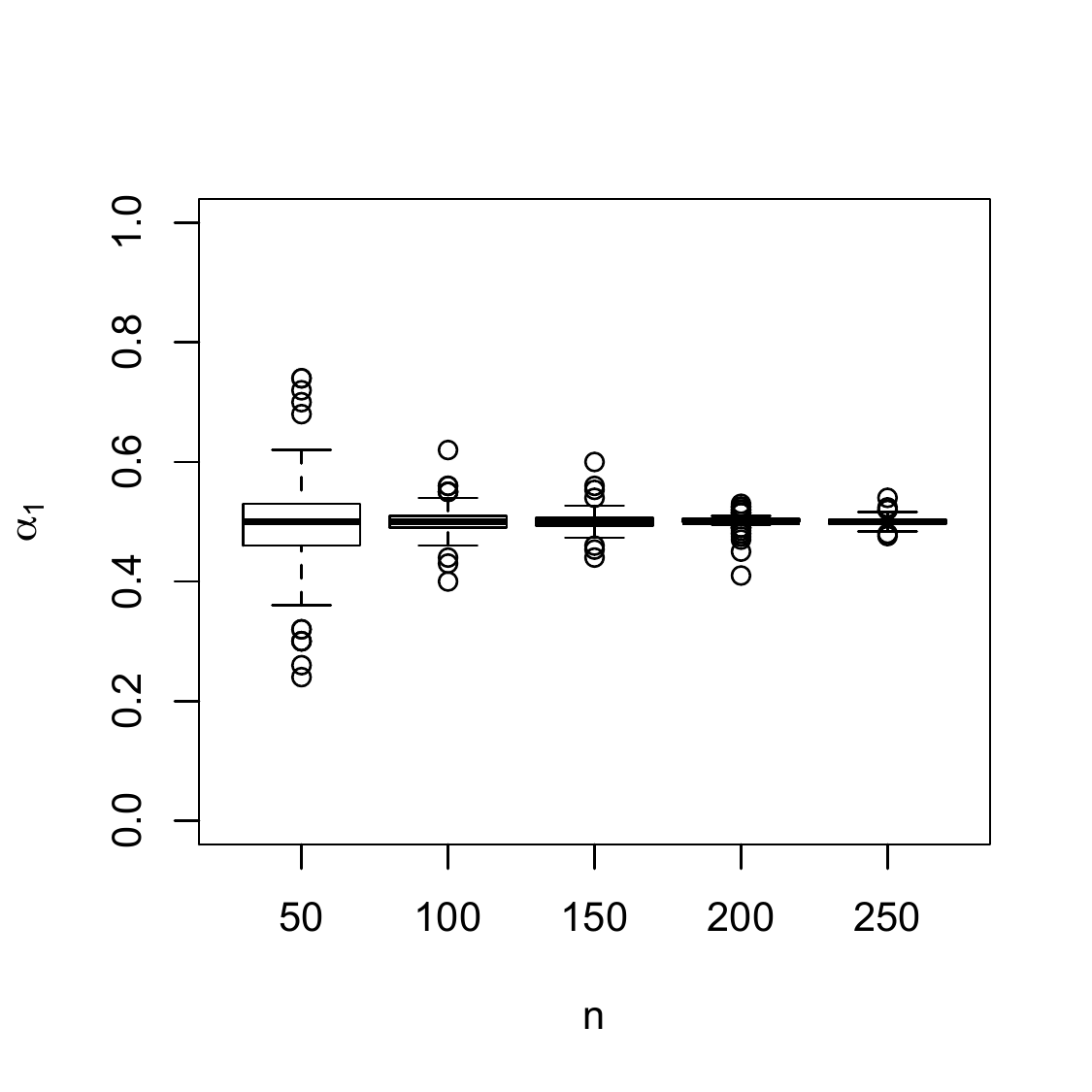}
\includegraphics[scale=0.35,trim=1cm 1cm 0 0,clip=true]
{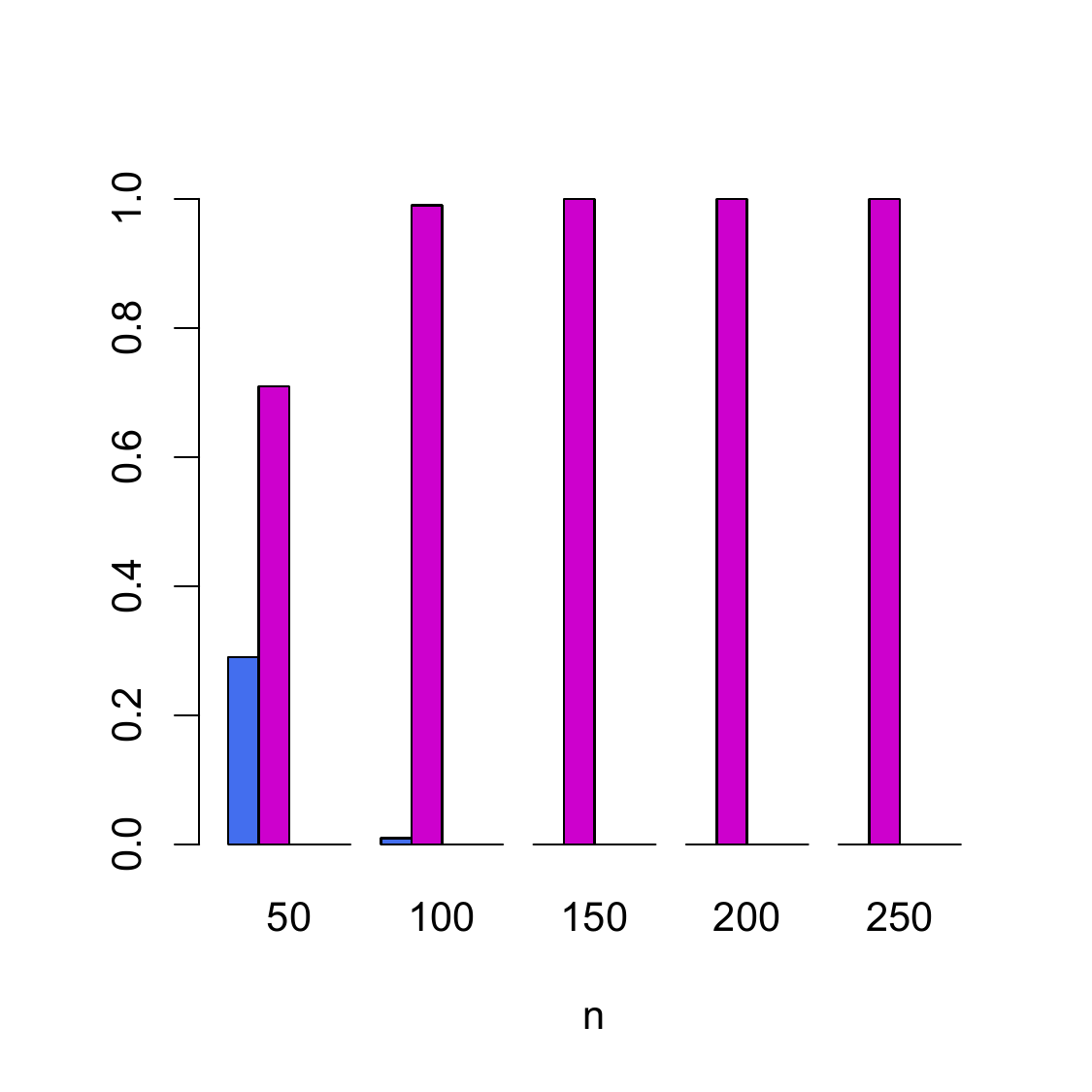}
\includegraphics[scale=0.35,trim=1cm 1cm 0 0,clip=true]
{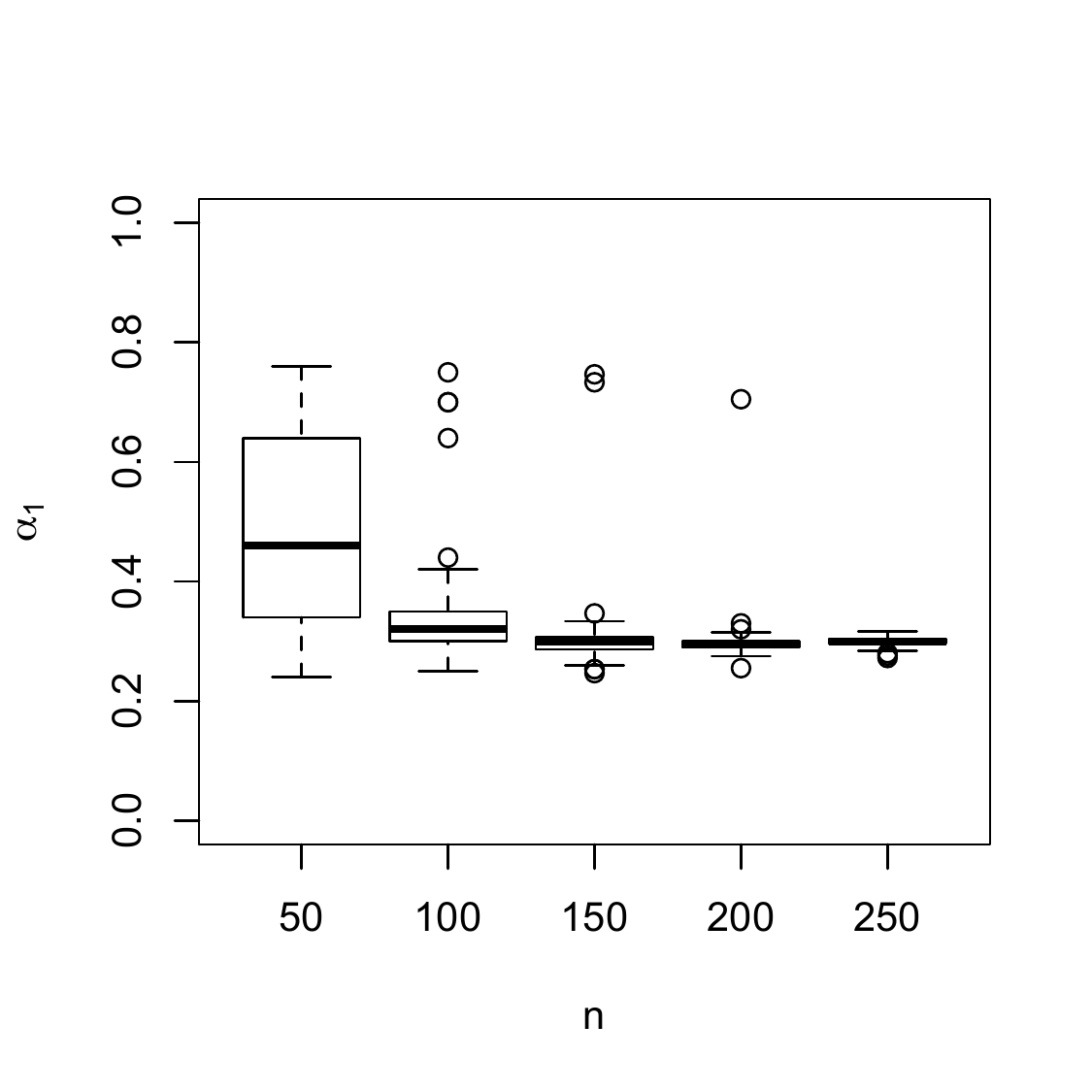}
\includegraphics[scale=0.35,trim=1cm 1cm 0 0,clip=true]
{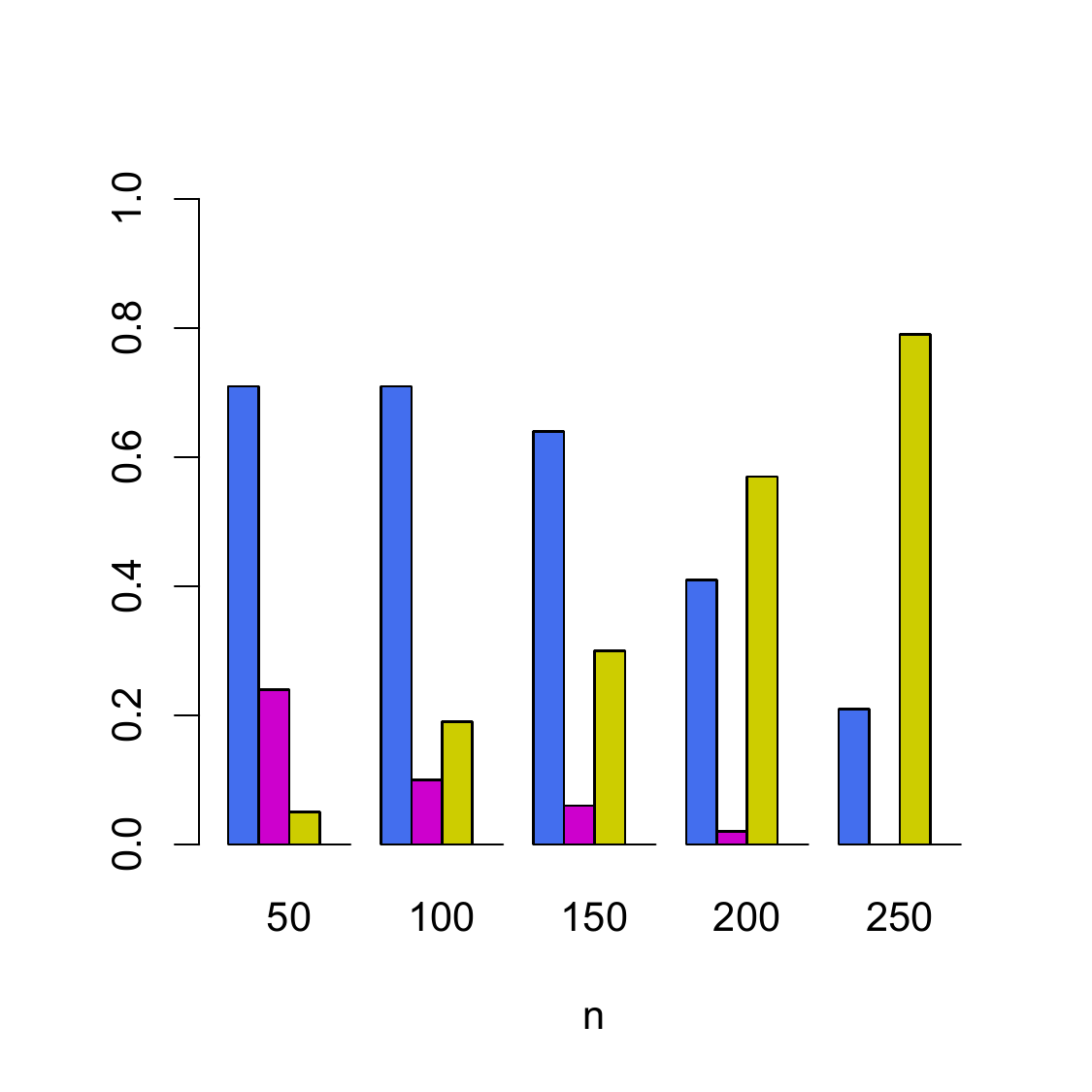}
\caption{As in Figure~\ref{fig:sim1}, but with model (3) for
  covariance matrix $\Sigma_{ij}=1-0.8\cdot \1_{\{i\neq j\}}$ for all
  $i,j$. Top: global estimator 
  \eqref{hatalpha} using DPA; Bottom: BS-algorithm. Left two panels: two
  segments model; Right two panels: three segments model. The
  barplots correspond to the relative
  frequencies that the algorithm gave a
  estimated single segment model (blue), a two segments model (magenta), a
  three segments model (yellow) or a four or more segments model (green).}
\label{fig:sim3}
\end{figure}
As can be seen from Figures~\ref{fig:sim1}--\ref{fig:sim3}, the
performances of the exact dynamic programming algorithm 
(DPA) and the binary segmentation algorithm (BS) are similar for
  larger sample size $n$. For small sample size $n$, the DPA method is
  superior to the BS algorithm
in the 
three segments model and they both perform well in the two segments model.  
But the computational times of the algorithms are very different, as
illustrated in Figure~\ref{fig:time}, where we show the mean time on 100 runs of
each algorithm for each sample size. As expected, the BS algorithm scales
much better with respect to sample size $n$.  
\begin{figure}[!htb]
\includegraphics[scale=0.6]{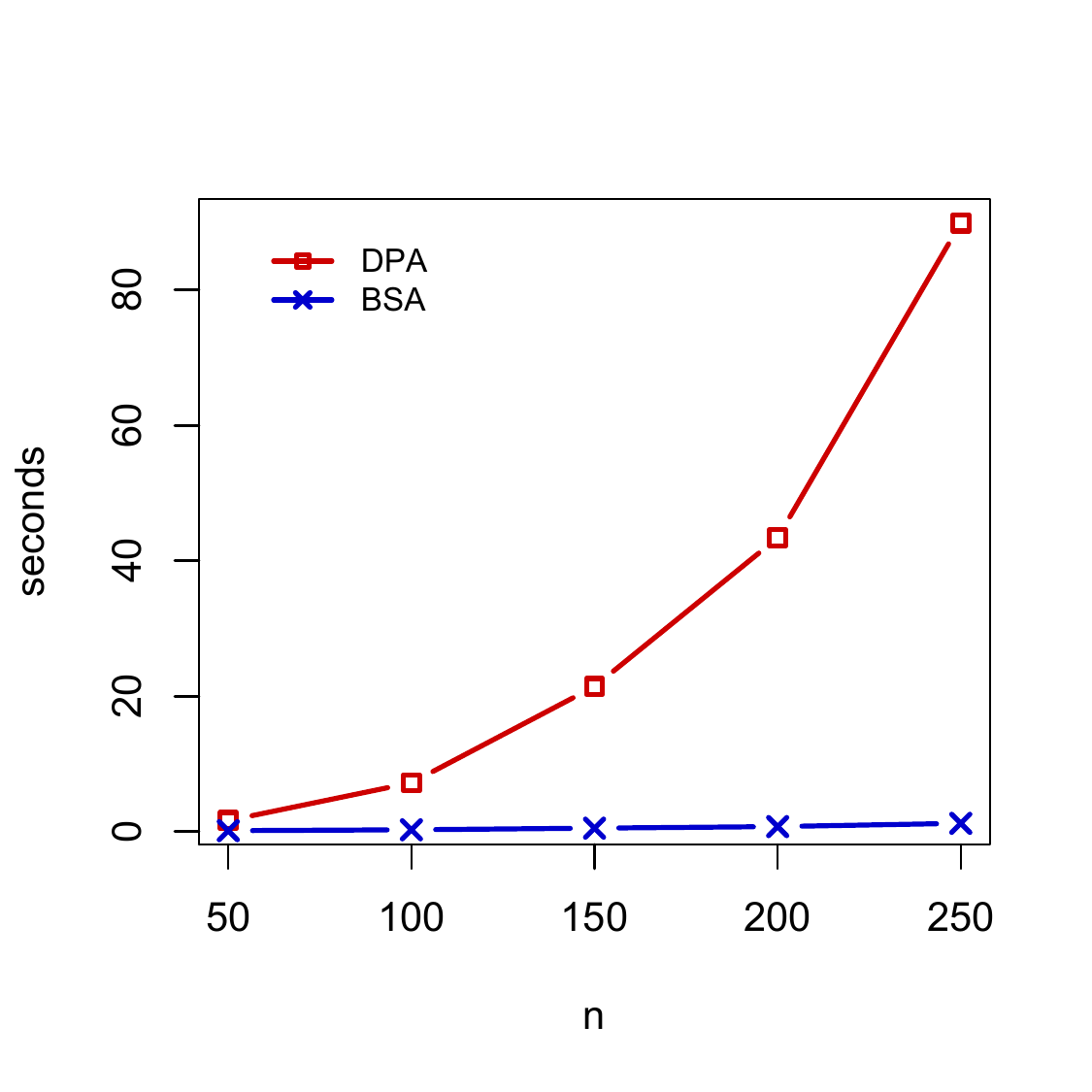}
\includegraphics[scale=0.6]{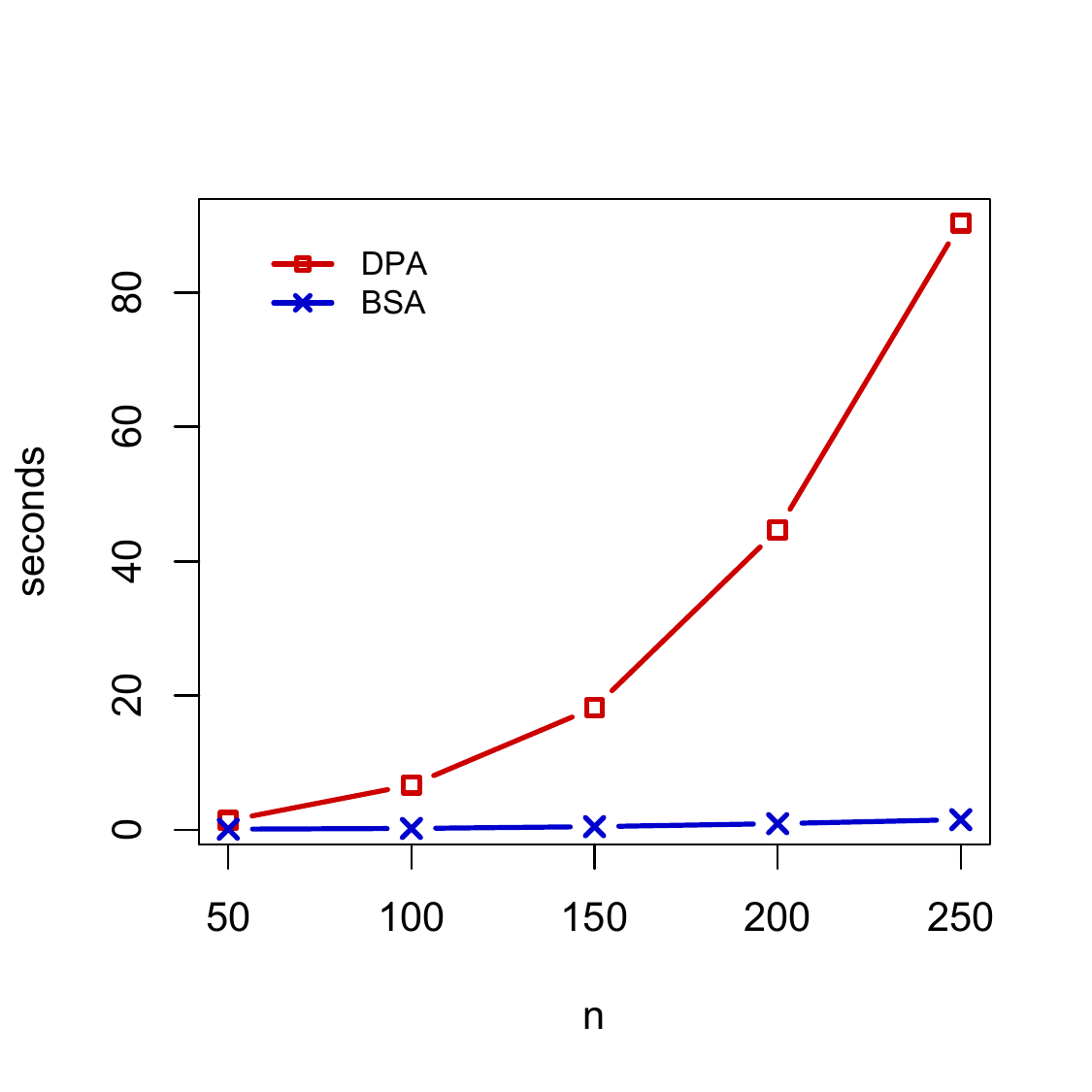}
\caption{Average computation time of the Dynamic Programming algorithm (DPA),
  giving the exact global minimum of \eqref{hatalpha}  and the  Binary
  Segmentation algorithm (BSA), providing an approximation to the global
  minimum as a function of sample size. Left panel: one change-point model 
  with  $\alpha^0=(0,0.5,1)$; right panel: two change-point model 
  with  $\alpha^0=(0,0.3,0.7,1)$.}
\label{fig:time}
\end{figure}

\section{Application to real data}\label{real-data}


We consider the ``communities and crime data'' (by M. Redmond) from the UCI
Machine Learning Repository 
\url{http://archive.ics.uci.edu/ml/datasets/Communities+and+Crime+Unnormalized#}.  
It comprises information from different communities in the U.S. and combines
socio-economic data, from the 1990 US Census and the 1990 US Law
Enforcement Management and Administrative Statistics Survey, and crime
data from the 1995 US FBI Uniform Crime Report. 

Besides specific information to identify the community (name, state, etc.)
the dataset comprises 125 predictive variables (population, mean people per
household, etc.) and 18 crime indices (number of murders per 100K
population, number of violent  crimes per 100K population, etc.).  
After removing all communities with missing values, we obtained a dataset
with $n=319$ communities and  $p=125$ covariates. We
selected as response of interest the 
(scaled) number of murders per 100K population in 1995. We assigned to each
community a number identifying its region in the following way: 1-South,
2-West, 3-Midwest, 4-Northeast (these regions are defined by the United
States Census Bureau) and then ordered the sample by regions (with the
original order from the dataset within every region). 

As a cross-validation procedure, we selected a
sub-sample of 160 communities  with indices $\{2i-1\colon i=1,\dotsc,
160\}$ and a test sample comprising the communities with indices
$\{2i\colon i=1,\dotsc, 159\}$. For a fixed $\delta=0.1$, $\lambda\in [0.001,2]$
and  
$k\in \{1,\dotsc, 10\}$ we computed the estimated $\alpha$ vector with
$\ell(\alpha)=k$ given by the exact dynamic programming algorithm over the
training dataset (i.e., we used the equivalent tuning parameter $k$ instead of
$\gamma$) and we then computed the residual sum of squares over the
test dataset.
The results are summarized in Figure~\ref{crime2norm}: the DPA (on top)
attains the minimum at $\lambda=0.051$ and $k=2$; the BSA attains the minimum at
$\lambda=0.073$ and $k=4$. We see that a one segment model is
  clearly out-performed with $k \ge 2$, with both algorithms DPA and
  BSA. We also see that the residual sum of squares curves for $k = 2$ or
  $k=3$ are essentially the same for both DPA and BSA. Thus $k  =2$ or
  $k=3$ almost leads to a minimum for the
  BSA, implying that $k \in \{2,3\}$ seems plausible for both
  methods. This finding makes sense: if we assume that the data is
    homogeneous within each region, there would be at most 4 segments.

\begin{figure}[!htb]
\center
\includegraphics[scale=0.72]{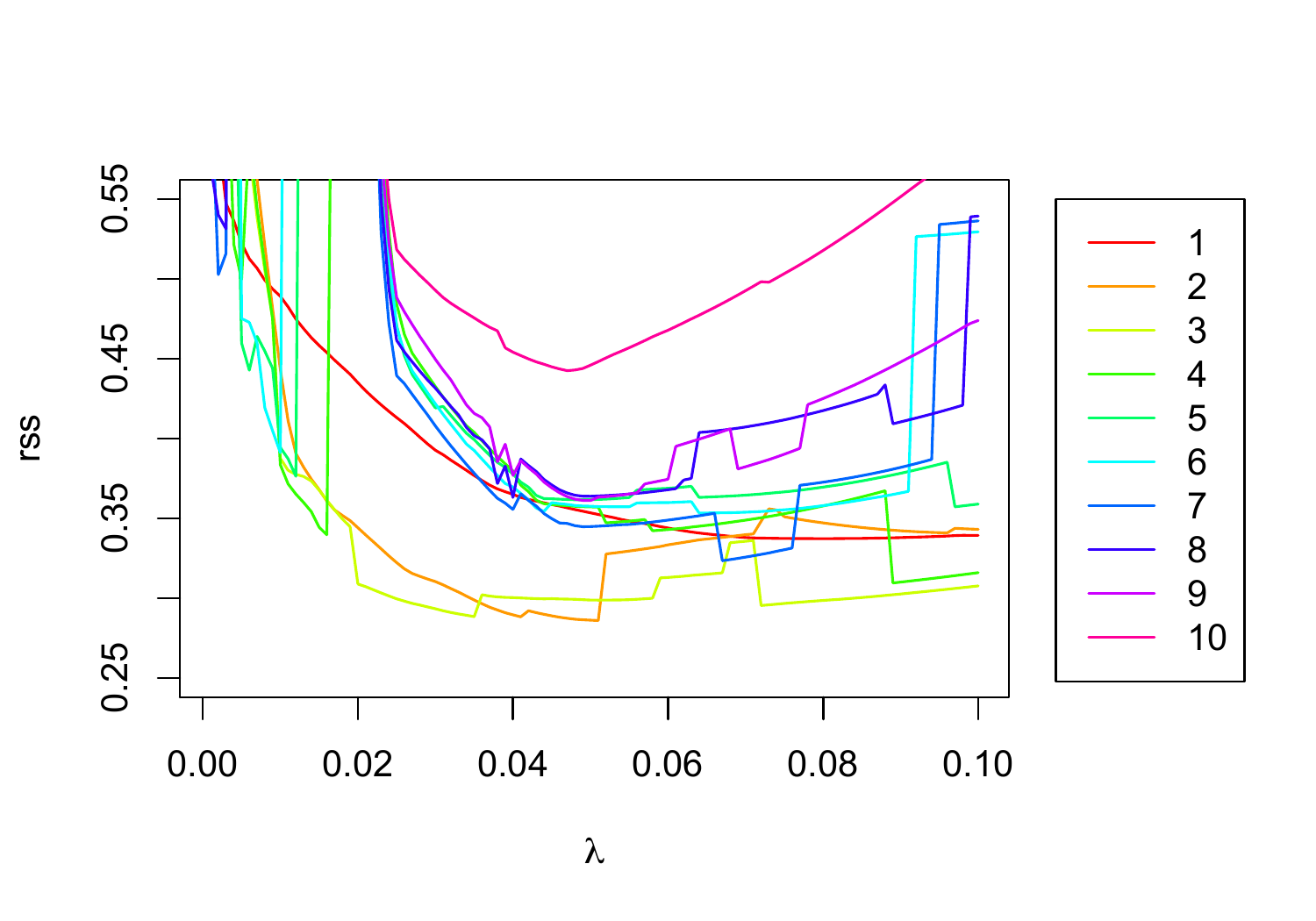}
\includegraphics[scale=0.72]{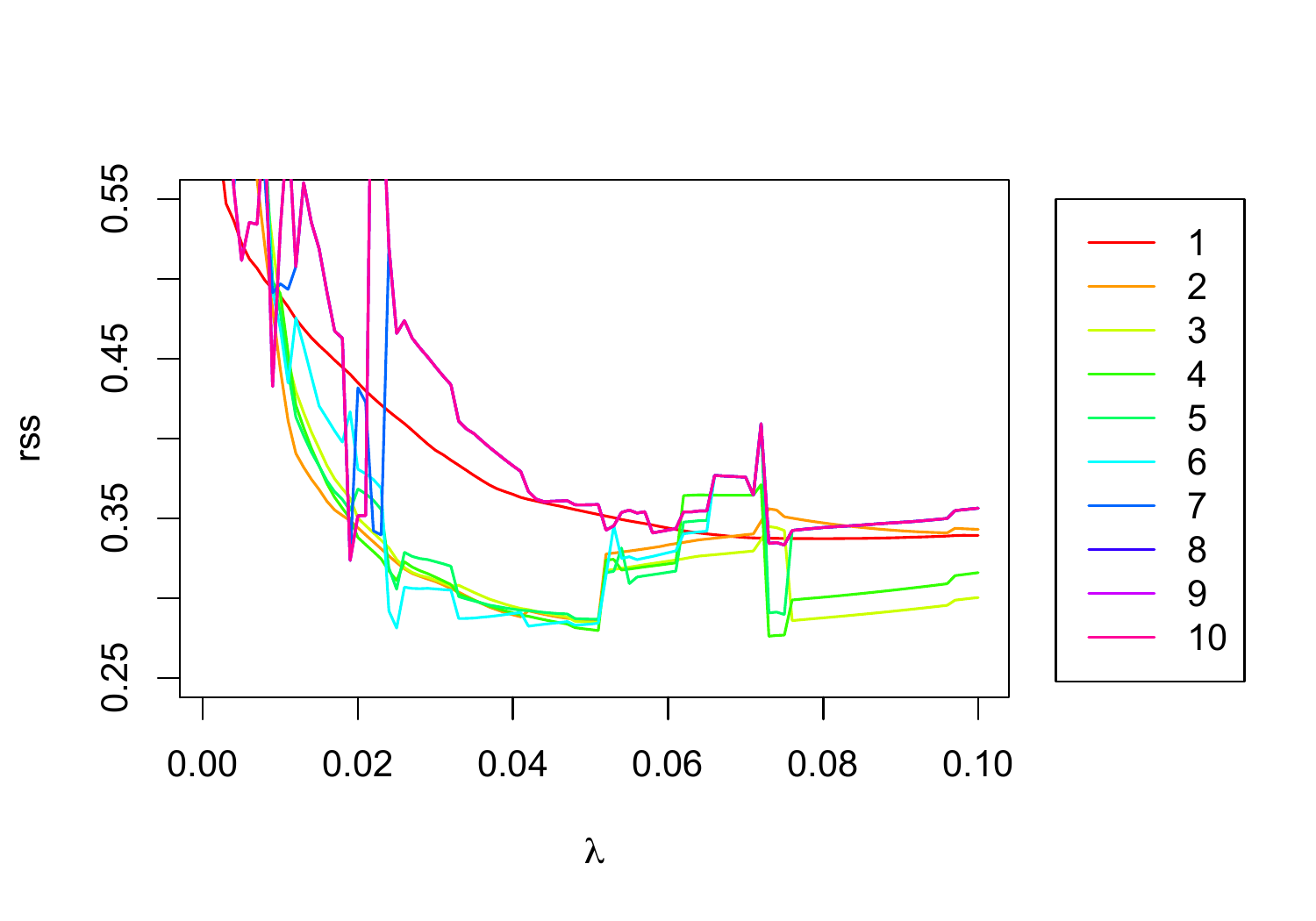}
\caption{Residual sum of squares computed on the test sample of
  the ``communities and crime'' dataset for different values of $\lambda\in
  [0.001,2]$ (only the range $[0.001,0.1]$ is shown) and
  number of segments between $k \in \{1,\ldots ,10\}$ (different
  lines). Top: Dynamic programming algorithm (DPA); bottom: Binary search
  algorithm (BSA). For DPA (top), he minimum is 
  attained at $\lambda=0.051$ for a model with 2 segments, and for BSA
  (bottom) at $\lambda=0.073$ for a
  model with 4 segments. For $k  \in \{2,3\}$, both
  algorithms DPA and BSA lead to essentially the same curves of the
  residual sum of squares as a function of $\lambda$.} 
\label{crime2norm} 
\end{figure}

\section{Conclusions}

Large-scale data is often exposed to heterogeneity: we consider here the
problem of detecting structural changes in the regression 
parameter of a high-dimensional linear model. We propose a regularized
residual sum of squares estimator, mainly using $\ell_1$-norm
regularization. The estimator can be either computed by dynamic programming
or, as mainly advocated in this work, it can be greedily approximated by a
computationally efficient scheme using recursive binary segmentation (BS
algorithm). Despite that the BS algorithm will not compute the regularized
residual sum of squares, we prove here the same theoretical properties for both methods: namely,
the consistency for the true number of segments (which is allowed to grow
asymptotically) and an oracle inequality implying a
fast convergence rate for prediction and parameter estimation. Thus, the
computationally much more efficient BS algorithm has the same theoretical
guarantees as the estimator based on a global optimum of the regularized
residual sum of squares. We illustrate the methods on simulated
as well as on a real dataset.   

\appendix

\section{Lasso estimator on a sub-interval}\label{l1}

In this section we present non-asymptotic oracle inequalities for the
estimators $\hat\beta_{(u,v]}$ in \eqref{lassobeta} that will be essential to
derive Theorem~\ref{modsel}. 

Given $k\in \N$ we denote by $I^m$ the set of intervals 
\[
I^m = \{(u,v]\subset(0,1]\colon (v-u)n\geq m \text{  and } \,un,vn\in\N\}\,.
\]
We can view the set $I^m$ as the collection of all possible sub-intervals
of the set $\{1,\dotsc, n\}$ with at least $m$ observations.

Given an interval $(u,v]\in I^1$ we  define the 
oracle $\beta^*_{(u,v]}$ by 
\begin{align}\label{betastar}
\beta^*_{(u,v]} \;&=\; \underset{\beta}{\arg\min}\;\; \|\bY_{(u,v]}- \bX_{(u,v]}\beta\|_{L^2(P)}^2\notag\\
&=\; \underset{\beta}{\arg\min}\;\;\;\E \|\bY_{(u,v]}-\bX_{(u,v]}\beta\|_2^2\,.
\end{align}
As $\beta^*_{(u,v]}$ is the minimizer of the above expression we have that the vector $\bX_{(u,v]}\beta^*_{(u,v]}$ represents the best approximation to $\bY_{(u,v]}$ in the linear subspace generated by the columns of $\bX_{(u,v]}$, with the inner product inherited from
the  $L^2(P)$ space.

For any $(u,v]\in I^1$,  define 
\begin{equation}\label{epsstar}
\epsilon_{(u,v]}^*=\bY_{(u,v]}-\bX_{(u,v]}\beta^*_{(u,v]}\,,
\end{equation}
and   let the set $\T_0$ be given by
\begin{equation}\label{t0}
\T_0 = \Bigl\{ \; \underset{(u,v]\in I^1}{\max}\;\underset{1\leq j\leq p}{\max}\; 2\bigl|  \be^{*,T}_{(u,v]}\bX_{(u,v]}^{(j)}\bigr|/{n} \,\leq\, 
\lambda_0 \;\Bigr\} \,.
\end{equation}

Now define the set  $\T_1$ by 
\begin{equation}\label{t1}
\T_1 = \Bigl\{ \underset{(u,v]\in I^1}{\max} \,\| \hat\Sigma_{(u,v]}-(v-u)\Sigma\|_\infty \leq 
\lambda_1  \Bigr\} \,,
\end{equation}
where 
\[
\hat\Sigma_{(u,v]} = \bX_{(u,v]}^T\bX_{(u,v]}/{n}\,.
\]

The following theorem shows oracle inequalities for the estimator \eqref{lassobeta}
on the sub-interval $(u,v]\subset(0,1]$. 

\begin{theorem}\label{mixture}
If Assumption 3 holds then on  the set $\T_0\cap\T_1$, with  $2\lambda_0 \leq \lambda\sqrt{\delta}$ and $s_*\lambda_1\leq \frac{\phi_*^2}{32}$ we have that  
\[
\|\bX_{(u,v]}(\hat\beta_{(u,v]}-\beta^*_{(u,v]}) \|_2^2/n + \lambda\sqrt{\max(v-u,\delta)}\|\hat\beta_{(u,v]} -\beta^*_{(u,v]}\|_1 \;\leq\; \frac{8\lambda^2\max(v-u,\delta)s_*}{(v-u)\phi_*^2}\,.
\] 
for all $(u,v]\in I^1$.
\end{theorem}

\medskip
\begin{remark}
Observe that the bound on Theorem~\ref{mixture} is uniform on the set
$I^{\delta n }$ with 
\[
\max_{(u,v]\in I^{\delta n}}
\left(\|\bX_{(u,v]}(\hat\beta_{(u,v]}-\beta^*_{(u,v]}) \|_2^2/n +
  \lambda\sqrt{v-u}\|\hat\beta_{(u,v]} -\beta^*_{(u,v]}\|_1 \right) \;\leq\; \frac{8\lambda^2s_*}{\phi_*^2}\,.
\]
\end{remark}

\medskip
\begin{corollary}\label{cormixture}
Suppose Assumptions 1-3 hold. Given $t>0$ and $\delta>0$, suppose 
 the regularization parameter $\lambda$ satisfies 
\[
\lambda \;\geq\; 40 t \sigma K_X\sqrt{ \frac{\log(np)}{\delta n}}\,.
\] 
Then if     
\begin{align*}
s_* & <  \frac{\lambda_1^{-1}\phi_*^2}{32}\,,\quad  \text{ with }\quad\lambda_1= 10tK_X^2\sqrt{ \frac{\log(np)}{n}}
\end{align*}
we have, with probability at least $1-  2/t^2$,  that\\[-10pt]
\[
\max_{(u,v]\in I^{\delta n}} \left(\|\bX_{(u,v]}(\hat\beta_{(u,v]}-\beta^*_{(u,v]}) \|_2^2/n + \lambda\sqrt{v-u}\|\hat\beta_{(u,v]} -\beta^*_{(u,v]}\|_1 \right)
\;\leq\; \frac{8\lambda^2s_*}{\phi_*^2}\,.
\] 
\end{corollary}

\section{Proofs}\label{sec.proofs} 

In this section we present the proofs of the theoretical results in this paper. In the first subsection we prove the oracle inequalities for the Lasso estimator on a subinterval, stated in Theorem~\ref{mixture} and Corollary~\ref{cormixture}. In the second subsection we prove the consistency of the change point estimators, stated in Theorems~\ref{modsel} and \ref{modsel2}.

\subsection{Oracle inequalities for the Lasso estimator}\label{subl1}

We first prove a result  about the compatibility condition.

\begin{lemma}\label{lem-comp}
Suppose Assumption 3 holds. Then on $\T_1$, if $\lambda_1$  satisfies $s_*\lambda_1 \leq \frac{\phi_*^2}{32}$, with $s_*$ the cardinality of $S_*$,  we have that for all $(u,v]\in I^1$ and all  $\beta\in \R^p$ that satisfy $\|\beta_{S_*^c}\|_1\leq 3\|\beta_{S_*}\|_1$ it holds that
\begin{equation*}
\|\beta_{S_*}\|_1^2 \;\leq\;  \frac{2(\beta^T\hat\Sigma_{(u,v]} \beta ) s_*}{(v-u)\phi_*^2}
\,.
\end{equation*}
\end{lemma}

\begin{proof}
First note that  by Assumption~3, for any $(u,v]\in I^1$ we have 
\begin{equation*}
\|\beta_{S_*}\|_1^2 \;\leq\;  \frac{(\beta^T (v-u)\Sigma\beta ) s_*}{(v-u)\phi_*^2}
\end{equation*}
for all $\beta\in \R^p$ that satisfy $\|\beta_{S_*^c}\|_1\leq 3\|\beta_{S_*}\|_1$. Therefore the matrix $(v-u)\Sigma$ satisfies the compatibility condition for the set $S_*$ with constant $\sqrt{(v-u)}\phi_*$. 
Now, by \cite[Corollary 6.8]{peter-sara-book} we have that if 
$s_*\lambda_1\leq \frac{\phi_*^2}{32}$, the compatibility condition   also holds for the set $S_*$ and the matrix $\hat\Sigma_{(u,v]}$ instead of $(v-u)\Sigma$, with $\phi_{\hat\Sigma_{(u,v]}}^2\geq (v-u)\phi_*^2/2$. That means that for all $\beta\in \R^p$ that satisfy $\|\beta_{S_*^c}\|_1\leq 3\|\beta_{S_*}\|_1$ it holds that
\begin{equation*}
\|\beta_{S_*}\|_1^2 \;\leq\; \frac{(\beta^T\hat\Sigma_{(u,v]} \beta ) s_*}{\phi_{\hat\Sigma_{(u,v]}}^2}\;\leq\;  \frac{2(\beta^T\hat\Sigma_{(u,v]} \beta ) s_*}{(v-u)\phi_*^2}
\,.\qedhere
\end{equation*}
\end{proof}

We  now prove a basic lemma  that can be derived  straightforward from \cite[Lemma~6.3]{peter-sara-book}. 

\begin{lemma}\label{lem6.3}
 On $\T_0$ with $2\lambda_0\leq \lambda\sqrt\delta$ we have that 
 \begin{align*}
2\|\bX_{(u,v]}(\hat\beta_{(u,v]} -  \beta^*_{(u,v]})&\,\|_2^2/n + \lambda\sqrt{\max(v-u,\delta)}\|\hat\beta_{(u,v],S_*^c}\|_1\\
&\leq\;
3\lambda\sqrt{\max(v-u,\delta)} \|\hat\beta_{(u,v],S_*}-\beta^*_{(u,v],S_*}\|_1
\end{align*}
for all $(u,v]\in I^1$.
\end{lemma}

\begin{proof}
Fix a interval $(u,v]\in I^1$ and denote by $\tilde\lambda= \lambda\sqrt{\max(v-u,\delta)}$. The Basic Inequality  in \cite[Lemma~6.1]{peter-sara-book}, derived directly from the definition \eqref{lassobeta} 
gives
\begin{align*}
\|\bX_{(u,v]}(\hat\beta_{(u,v]} -& \beta^*_{(u,v]})\|_2^2/n + \tilde\lambda
\|\hat\beta_{(u,v]}\|_1\\
&\leq\;
 2 (\be^*_{(u,v]})^T \bX_{(u,v]}(\hat\beta_{(u,v]}-\beta^*_{(u,v]}) /{n} +  \tilde\lambda\|\beta^*_{(u,v]}\|_1\,.
\end{align*}
Now, on $\T_0$ and using $\tilde\lambda\geq 2\lambda_0$ we have  
\begin{align}\label{eqpr1}
2\|\bX_{(u,v]}(\hat\beta_{(u,v]} -& \beta^*_{(u,v]})\|_2^2/n + 2 \tilde\lambda\|\hat\beta_{(u,v]}\|_1\notag\\
&\leq\;\tilde\lambda \|\hat\beta_{(u,v]}-\beta^*_{(u,v]} \|_1+ 2\tilde\lambda\|\beta^*_{(u,v]}\|_1\,.
\end{align}
By using the triangle inequality we obtain
\begin{align*}
\|\hat\beta_{(u,v]}\|_1 &= \|\hat\beta_{(u,v],S_*}\|_1 + \|\hat\beta_{(u,v],S_*^c}\|_1\\
&\geq  \|\beta^*_{(u,v],S_*}\|_1  -  \|\hat\beta_{(u,v],S_*}-\beta^*_{(u,v],S_*}\|_1 +  \|\hat\beta_{(u,v],S_*^c}\|_1\,.
\end{align*}
On the other hand we also have that 
\[
  \|\hat\beta_{(u,v]}-\beta^*_{(u,v]}\|_1 =   \|\hat\beta_{(u,v],S_*}-\beta^*_{(u,v],S_*}\|_1
  +   \|\hat\beta_{(u,v],S_*^c}\|_1\,.
\]
By plugin-in these last expressions in \eqref{eqpr1} we finish the proof of Lemma~\ref{lem6.3}.
\end{proof}

We now prove the main result in Appendix~\ref{l1}, given by Theorem~\ref{mixture} and Corollary~\ref{cormixture}.

\begin{proof}[Proof of Theorem~\ref{mixture}]
The proof follows the same lines of reasoning as Theorem~6.1 in \cite{peter-sara-book}.
As before, fix a interval $(u,v]\in I^1$ and denote by $\tilde\lambda= \lambda\sqrt{\max(v-u,\delta)}$. Then on $\T_0\cap\T_1$, if $2\lambda_0\leq \lambda\sqrt{\delta}\leq \tilde\lambda$ we have, by Lemma~\ref{lem6.3}  that 
\begin{align*}
2\| \bX_{(u,v]}(\hat\beta_{(u,v]}-&\,\beta^*_{(u,v]})\|_2^2/n + \tilde\lambda\|\hat\beta_{(u,v]} -\beta^*_{(u,v]}\|_1 \\
&=\;2\| \bX_{(u,v]}(\hat\beta_{(u,v]}-\beta^*_{(u,v]})\|_2^2/n + \tilde\lambda\|\hat\beta_{(u,v],S_*^c}\|_1\\
&\qquad + \tilde\lambda\|\hat\beta_{(u,v],S_*} -\beta^*_{(u,v],S_*}\|_1\\
&\leq  \;4 \tilde\lambda\|\hat\beta_{_{(u,v]},S_*} -\beta^*_{_{(u,v]},S_*}\|_1\,.
\end{align*}
Now, if $s_*\lambda_1\leq \frac{\phi_*^2}{32}$ then  by Lemma~\ref{lem-comp} and the inequality $4ab\leq a^2+ 4b^2$
we can bound above the right hand side of the last expression by
\begin{align*}
4 \tilde\lambda \sqrt{2s_*}\|&\bX_{(u,v]}(\hat\beta_{(u,v]}-\beta^*_{(u,v]})\|_2/\sqrt{(v-u)n}\,\phi_*
\\
&\leq\; \|\bX_{(u,v]}(\hat\beta_{(u,v]}-\beta^*_{(u,v]})\|_2^2/n + 4\, \frac{\tilde\lambda^2 2s_*}{(v-u)\phi_*^2 }
\end{align*}
and this concludes the proof. 
\end{proof}

\begin{proof}[Proof of Corollary~\ref{cormixture}]
The result follows  by combining the result in Theorem~\ref{mixture}  with Lemmas~\ref{lemma-t0} and \ref{lemma-t1}.
\end{proof}

\subsection{Proofs of Theorems \ref{modsel} and \ref{modsel2}}\label{proofmain}

In this subsection we present the proof of Theorems~\ref{modsel} and ~\ref{modsel2}
and all the auxiliary results. 

We need some extra notation. Given the values $u\leq \eta \leq v$ and vectors $\beta,\beta^{(1)}$ and $\beta^{(2)}\in\R^p$ we can write 
\begin{align}\label{eqmain}
\|\bY_{(u,v]}-\bX_{(u,v]}\beta\|_2^2 &\;=\; \|\bY_{(u,\eta]}-\bX_{(u,\eta]}\beta^{(1)}\|_2^2  +  \|\bY_{(\eta,v]}-\bX_{(\eta,v]}\beta^{(2)}\|_2^2 \notag\\ 
&\quad + \|\bD_{(u,\eta]}(\beta,\beta^{(1)})\|_2^2  -  2\,\hat\be_{(u,\eta]}^T(\beta^{(1)}) \bD_{(u,\eta]}(\beta,\beta^{(1)}) \\
&\quad+ \| \bD_{(\eta,v]}(\beta,\beta^{(2)})\|_2^2  -  2\,\hat\be_{(\eta,v]}^T(\beta^{(2)}) \bD_{(\eta,v]}(\beta,\beta^{(2)})\,,\notag
\end{align}
where $\bD_{(u,\eta]}(\beta,\beta^{(1)}) = X_{(u,\eta]}(\beta- \beta^{(1)})$,
$\bD_{(\eta,v]}(\beta,\beta^{(2)}) = X_{(\eta,v]}(\beta- \beta^{(2)})$,
$\hat\epsilon_{(u,\eta]}(\beta^{(1)}) = \bY_{(u,\eta]}-\bX_{(u,\eta]}\beta^{(1)}$
and $\hat\epsilon_{(\eta,v]}(\beta^{(2)}) = \bY_{(\eta,v]}-\bX_{(\eta,v]}(\beta^{(2)})$.

\medskip 
We can now prove the following result. 

\begin{lemma}\label{lem1}
Suppose $k_0>1$ and that Assumptions 1-4 hold. Then on $\T_0\cap \T_1$,  if   $u<\alpha^0_j<v$ for some $j=1,\dotsc,k_0-1$ and $s_*\lambda_1 \leq \frac{\phi_*^2}{32}$ 
we have
\begin{align*}
\frac{\|\bD_{(u,\alpha_j^0]}(\beta^*_{(u,v]},\beta^*_{(u,\alpha^0_j]})\|_2^2}n+ \frac{\|\bD_{(\alpha_j^0,v]}(\beta^*_{(u,v]},\beta^*_{(\alpha^0_j,v]})\|_2^2}n
\geq  \frac{\min(\alpha_j^0-u,v-\alpha^0_j)  m_*^2 \phi_*^2 s_*}{8}\,.
\end{align*}
\end{lemma}

\begin{proof}
Let $j=1,\dotsc, k_0-1$ and $u<\alpha^0_j<v$ such that $(u,\alpha^0_j],(\alpha^0_j,v]\in I^1$. Denote by  $\eta=\alpha_j^0$.  
By definition 
 we have   
\begin{align}\label{eq1}
\|\bD_{(u,\eta]}(\beta^*_{(u,v]},\beta^0(j-1))\|_2^2\;  = &\; \|\bX_{(u,\eta]}(\beta^*_{(u,v]}-\beta^0(j-1))\|_2^2
\end{align}
and a similar expression for $\|\bD_{(\eta,v]}(\beta^*_{(u,v]},\beta^0(j))\|_2^2$. 
By Assumptions~3 and Lemma~\ref{lem-comp} we have 
\begin{align*}
&\frac{\|\bX_{(u,\eta]}(\beta^*_{(u,v]}-\beta^*_{(u,\eta]}\|^2_2}{n}+ \frac{\|\bX_{(\eta,v]}(\beta^*_{(u,v]}-\beta^*_{(\eta,v]}\|^2_2}{n}\\
& \qquad\;\geq\; \frac{(\eta-u)\|\beta^*_{(u,v]}-\beta^*_{(u,\eta]}\|_1^2\phi_*^2}{2s_*} +
 \frac{(v-u)\|\beta^*_{(u,v]}-\beta^*_{(\eta,v]}\|_1^2\phi_*^2}{2s_*}
\end{align*}
Now observe that 
\[
(v-u)\beta^*_{(u,v]} = (\eta-u)\beta^*_{(u,\eta]} + (v-\eta)\beta^*_{(\eta,v]}
\]
then 
\[
\beta^*_{(u,v]}-\beta^*_{(u,\eta]} = \Bigl(\frac{v-\eta}{v-u} \Bigr)\Bigl(\beta^*_{(\eta,v]} -  \beta^*_{(u,\eta]} \Bigr)
\]
and by Assumption~4 and Lemma~\ref{convex} we have 
\[
\|\beta^*_{(u,v]}-\beta^*_{(u,\eta]}\|_1 = \frac{v-\eta}{v-u}\|\beta^*_{(\eta,v]} -  \beta^*_{(u,\eta]}\|_1\geq  \frac{(v-\eta)m_* s_*}{(v-u)} \,.
\]
Similarly we obtain
\[
\|\beta^*_{(u,v]}-\beta^*_{(\eta,v]}\|_1 = \frac{\eta-u}{v-u}\|\beta^*_{(\eta,v]} -  \beta^*_{(u,\eta]}\|_1\geq  \frac{(\eta-u)m_* s_*}{(v-u)} \,.
\]
Then 
\begin{align*}
 \frac{(\eta-u)\|\beta^*_{(u,v]}-\beta^*_{(u,\eta]})\|_1^2\phi_*^2}{2s_*} +
 \frac{(v-u)\|\beta^*_{(u,v]}-\beta^*_{(\eta,v]}\|_1^2\phi_*^2}{2s_*} \;\geq\; \frac{(\eta-u)(v-\eta)\phi_*^2}{2(v-u)s_*}
\end{align*}
and as $\max(\eta-u,v-\eta)/(v-u)\geq 1/2$ this concludes the proof.
\end{proof}

\begin{lemma}\label{lem2}
For any interval $(u,\eta]\in I^1$ and any $\beta\in\R^p$ we have 
 \[
2\, |\hat\be_{(u,\eta]}^T(\hat\beta_{(u,\eta]}) \bD_{(u,\eta]}(\beta,\hat\beta_{(u,\eta]})| /n \;\leq\; 
\lambda \sqrt{\eta-u}  \, \|\beta-\hat\beta_{(u,\eta]}\|_1\,.
 \]
 Additionally, on $\T_0$, if $2\lambda_0 \leq \lambda \sqrt\delta$ we have 
  \[
2\, |\hat\be_{(u,\eta]}^T(\beta^*_{(u,\eta]}) \bD_{(u,\eta]}(\beta,\beta^*_{(u,\eta]})| /n \;\leq\; 
\frac{\lambda\sqrt\delta}{2} \, \|\beta-\beta^*_{(u,\eta]}\|_1\,.
 \]
\end{lemma}

\begin{proof}
Note that 
\begin{align*}
\hat\be_{(u,\eta]}^T(\hat\beta_{(u,\eta]}) \bD_{(u,\eta]}(\beta,\hat\beta_{(u,\eta]}) &\;=\;
 (\bY_{(u,\eta]} - \bX_{(u,\eta]} \hat\beta_{(u,\eta]})^T \bX_{(u,\eta]} (\beta-\hat\beta_{(u,\eta]})\\
&\;=\; (\bX_{(u,\eta]}^T (\bY_{(u,\eta]} - \bX_{(u,\eta]} \hat\beta_{(u,\eta]})  )^T(\beta-\hat\beta_{(u,\eta]})\,.
\end{align*}
By \cite[Lemma~2.1]{peter-sara-book} we have that as $\hat\beta_{(u,\eta]}$ is the solution of \eqref{lassobeta} then 
\[
\|2 (\bX_{(u,\eta]}^T (\bY_{(u,\eta]} - \bX_{(u,\eta]} \hat\beta_{(u,\eta]})/n \|_\infty 
\;\leq\; \frac{\lambda \,(\eta-u)}{\sqrt{\max(\eta-u,\delta)}}\;\leq\; \lambda \sqrt{\eta-u} \,.
\]
Therefore
\begin{align*}
2\,|\hat\be_{(u,\eta]}^T \bD_{(u,\eta]}(\beta)|/n &\;\leq\;
\|2 (\bX_{(u,\eta]}^T (\bY_{(u,\eta]} - \bX_{(u,\eta]} \hat\beta_{(u,\eta]})/n \|_\infty \|\beta-\hat\beta_{(u,\eta]}\|_1\\
&\;\leq\; \lambda\sqrt{\eta-u} \|\beta-\hat\beta_{(u,\eta]}\|_1\,.\qedhere
\end{align*}
The bound for $2\, |\hat\be_{(\eta,v]}^T(\beta^*_{(u,\eta]}) \bD_{(\eta,v]}(\beta,\beta^*_{(u,\eta]})| /n$ is obtained analogously, by noting that on $\T_0$, if $2\lambda_0\leq \lambda\sqrt\delta$ we have 
\[
\|2 (\bX_{(u,\eta]}^T (\bY_{(u,\eta]} - \bX_{(u,\eta]}\beta^*_{(u,\eta]})/n \|_\infty \;=\; \max_{j=1,\dotsc, p} 2|\be^{*,T}_{(u,\eta]}\bX^{(j)}_{(u,\eta]}|\;\leq\;\lambda_0 \;\leq\; \frac{\lambda\sqrt\delta}{2}\,.\qedhere
\]
\end{proof}

\begin{lemma}\label{lem4}
Suppose Assumptions 1-4  hold and let 
\[
(u,v] \subset (\alpha_{j-1}^0-\lambda\sqrt\delta/d_*, \alpha^0_j+\lambda\sqrt\delta/d_*]\cap(0,1]
\]
for some $j=1,\dotsc,k_0$, with $(u,v]\in I^{\delta n}$ and $\lambda\sqrt\delta < r(\alpha^0)d_*$. 
Then on $\T_0\cap \T_1$, if $s_*\lambda_1 \leq \frac{\phi_*^2}{32}$ and  $2\lambda_0\leq \lambda\sqrt{\delta}$
we have
\begin{equation*}
\|\bX_{(u,v]}(\beta^0(j) - \hat\beta_{(u,v]})\|_2^2/n 
+ \lambda \sqrt{v-u} \|\beta^0(j) - \hat\beta_{(u,v]}\|_1\;\leq\; c(r) \lambda^2s_*\,,
\end{equation*}
where 
\[
c(r) = \Bigl( \frac{rK_X M_*}{d_*} + \frac{\sqrt{8}}{\phi_*}\Bigr)^2\,,\qquad r = \1\{u<\alpha^0_{j-1}\}+\1\{\alpha^0_{j}<v\}\,.
\] 
\end{lemma}
Note that Lemma \ref{lem4} is taking the bias $\|\beta^*_{(u,v]} -
  \beta^0(j)\|_1$ into account, as pointed out in the proof.
\begin{proof}
Observe that
\begin{align*}
\|&\bX_{(u,v]}(\beta^0(j) - \hat\beta_{(u,v]})\|_2^2/n 
+ \lambda \sqrt{v-u} \|\beta^0(j) - \hat\beta_{(u,v]}\|_1   \\ 
&\leq \|\bX_{(u,v]}(\beta^0(j) - \beta^*_{(u,v]})\|_2^2/n + 2 \|\bX_{(u,v]}(\beta^0(j)-\beta^*_{(u,v]})\|_2\|\bX_{(u,v]}(\beta^*_{(u,v]}-\hat\beta_{(u,v]})\|_2/n\\
& \, +  \|\bX_{(u,v]}(\beta^*_{(u,v]} - \hat\beta_{(u,v]})\|_2^2/n + \lambda \sqrt{v-u}\|\beta^*_{(u,v]} - \hat\beta_{(u,v]}\|_1 \\
&\, +   \lambda \sqrt{v-u}\|\beta^0(j) - \beta^*_{(u,v]}\|_1   \,.
\end{align*}
By Theorem~\ref{mixture} we obtain that 
\[
 \|\bX_{(u,v]}(\beta^*_{(u,v]} - \hat\beta_{(u,v]})\|_2^2/n  +  \lambda\sqrt{v-u} \|\beta^*_{(u,v]} - \hat\beta_{(u,v]}\|_1  \;\leq\; \frac{8 \lambda^2 s_*}{\phi_*^2}\,.
\]
On the other hand, if $\lambda\sqrt\delta < r(\alpha^0) d_*$ we have 
\begin{align*}\label{boundvec}
\|\beta^0(j) - \beta^*_{(u,v]}\|_\infty\;&\leq\; \frac{\max(\alpha_{j-1}^0-u,0)}{(v-u)} \|\beta^0(j)-\beta^0(j-1)\|_\infty\\
& \quad+  \frac{\max(v - \alpha_{j-1}^0,0)}{(v-u)} \|\beta^0(j+1)-\beta^0(j)\|_\infty\\
& \;\leq\; \frac{r M_*\lambda}{d_*\sqrt{v-u}}\,.
\end{align*}
Note that this inequality shows in particular that when $u$ is at distance at most $d$ of $\alpha^0_{j-1}$ and $v$ is at distance at most $d$ of $\alpha^0_j$ then the ``bias''  between $\beta^0(j)$ and $\beta^*_{(u,v]}$, measured by $\|\beta^0(j)-\beta^*_{(u,v]}\|_1$,  is of order $d M_*s_*$. 
Then, by using this bound  we also obtain that 
\begin{align*}
\|\bX_{(u,v]}&(\beta^0(j) - \beta^*_{(u,v]})\|_2^2/n + 2 \|\bX_{(u,v]}(\beta^0(j)-\beta^*_{(u,v]})\|_2\|\bX_{(u,v]}(\beta^*_{(u,v]}-\hat\beta_{(u,v]})\|_2/n\\
\leq &\;\;(v-u)s_* K_X^2\|\beta^0(j) - \beta^*_{(u,v]})\|_\infty^2\\
&\quad + 2 \sqrt{(v-u)s_*}K_X \|\beta^0(j) - \beta^*_{(u,v]}\|_\infty\|\bX_{(u,v]}(\beta^*_{(u,v]}-\hat\beta_{(u,v]})\|_2/\sqrt{n}\\
\leq &\;\;  \frac{r^2K_X^2 M_*^2\lambda^2 s_*}{d^2_*}+  \frac{ 2\sqrt{8} s_* rK_X M_* \lambda^2}{ d_*\phi_*}\,.
\end{align*}
By summing all the above bounds we obtain 
\begin{equation*}
\|\bX_{(u,v]}(\beta^0(j) - \hat\beta_{(u,v]})\|_2^2/n 
+ \lambda\sqrt{v-u} \|\beta^0(j) - \hat\beta_{(u,v]}\|_1
 \;\leq\;  c(r)\lambda^2s_*\,,
\end{equation*}
where  
\[
c(r) = \Bigl(\frac{rK_X M_*}{d_*} + \frac{\sqrt{8}}{\phi_*}\Bigr)^2 \,.\qedhere 
\]
\end{proof}

Now we can prove the main results in Section~\ref{mainsec} and \ref{algsec}.

\begin{proof}[Proof of Theorem~\ref{modsel}]
We begin by proving that points 1-3 hold on $\T_0\cap \T_1$ if  the conditions of the theorem are satisfied.   Then the  probability lower bound follows by combining this fact  with Lemmas~\ref{lemma-t0} and \ref{lemma-t1}. \\
To simplify notation, given a vector $\alpha$ as in \eqref{alph}, with 
$r(\alpha)\geq \delta$,  lets  denote by
$G(\alpha)$ the 
value of the function in \eqref{hatalpha} corresponding to the vector $\alpha$; that is
\[
G(\alpha) \;=\;  \sum_{j=1}^{\ell(\alpha)} L(I_j(\alpha),\hat\beta(j)) + \ell(\alpha)\gamma\,,
\]
where $\hat\beta(j)$ is given by \eqref{hatbeta} and $L=L_n$. 
By the identity in \eqref{lassobeta2} we have that  
\begin{equation}\label{G2}
G(\alpha) = \sum_{j=1}^{\ell{\alpha}} L(I_j(\alpha),\hat\beta_{I_j(\alpha)}) + \ell(\alpha) \gamma\,,
\end{equation}
where $\hat\beta_{I_j(\alpha)}$ is  the Lasso estimator for the interval $I_j(\alpha)$ given by \eqref{lassobeta}.  In the sequel we will also need the function $G(\alpha)$  defined on vectors $\alpha$ such that $r(\alpha)<\delta$; in these cases we consider the ``extended'' version \eqref{G2}, because $\hat\beta_{I_j(\alpha)}$ is defined in \eqref{lassobeta} even if $r_j(\alpha)<\delta$. 

For any $j=1,\dotsc,k_0$ denote by $\B(\alpha^0_j,\lambda\sqrt\delta /d_*)$ the ball of center $\alpha^0_j$ and radius $\lambda\sqrt\delta/d_*$. First we will show that on $\T_0\cap\T_1$, if the conditions of the theorem are satisfied we must have $\ell(\hat\alpha)=k_0$ and $\|\hat\alpha - \alpha^0\|_1\leq \lambda\sqrt\delta/d_*$, by showing that 
\begin{equation}\label{condalph}
\hat\alpha\cap \B(\alpha^0_j,\lambda\sqrt\delta/d_*) = \hat\alpha_j
\end{equation}
for all $j=1,\dotsc,k_0$. 
 To show this, we will prove that if $\hat\alpha$ does not satisfy \eqref{condalph} then there exists another ordered vector $\alpha=(\alpha_0,\dotsc,\alpha_k)$ such that $\alpha_0=0$, $\alpha_k=1$, $r(\alpha)\geq\delta$ and satisfying
\begin{equation}\label{ineq}
G(\alpha) \;<\; G(\hat\alpha)
\end{equation}
which contradicts the fact that $\hat\alpha$ minimizes \eqref{hatalpha}.
So, suppose that \eqref{condalph} does not hold, 
we distinguish two possible cases:

\begin{enumerate}
\item[(a)]  There exists some $i$, $1\leq i \leq \hat k -1$,  such that 
$\{\hat\alpha_{i-1},\hat\alpha_i,\hat\alpha_{i+1}\}\subset(\alpha^0_{j-1}-\lambda\sqrt\delta/d_*,\alpha^0_{j}+\lambda\sqrt\delta/d_*)\cap(0,1]$ for some  $j$, $1\leq j \leq k_0$.
\item[(b)]  $\hat\alpha\cap \B(\alpha^0_j,\lambda\sqrt\delta/d_*) = \emptyset$ for some $j=1,\dotsc,k_0-1$.
\end{enumerate}
In the case (a) define 
\[
\alpha= (\hat\alpha_0,\dotsc,\hat\alpha_{i-1},\hat\alpha_{i+1},\dotsc,\hat\alpha_{\ell(\hat\alpha)})
\]
so that $\ell(\alpha)=\ell(\hat\alpha)-1$. 
Denote by $J_1$ and $J_2$ the intervals
\[
J_1 = (\hat\alpha_{i-1},\hat\alpha_{i}]\,,\quad J_2 = (\hat\alpha_{i},\hat\alpha_{i+1}]
\]
and let $J$ denote their union $J=(\alpha_{i-1},\alpha_{i+1}]$.
Then
we obtain
\begin{align*}
G(\alpha) - G(\hat\alpha)\; 
&=\; L(J,\hat\beta_{J})   - L(J_1,\hat\beta_{J_1}) -L(J_2,\hat\beta_{J_2}) - \gamma\,.
\end{align*}
By the definition \eqref{hatbeta} we have that 
\begin{align*}
L(J,\hat\beta_{J}) \;\leq \;  
 &\;L(J,\beta^0(j)) + \lambda\sqrt{|J|} \|\beta^0(j)-\hat\beta_J\|_1
\end{align*}
and by the equality \eqref{eqmain} with $\beta=\beta^0(j)$,  
$\beta^{(1)} =\hat\beta_{J_1}$,  $\beta^{(2)} =\hat\beta_{J_2}$ and $\eta=\hat\alpha_i$ we have that 
\begin{align*}
L(J,\beta^0(j)) \;=\; &L(J_1,\hat\beta_{J_1}) +  L(J_2,\hat\beta_{J_2})+  \frac{\|\bD_{J_1}(\beta^0(j),\hat\beta_{J_1})\|_2^2}n -2 \frac{\hat\be_{J_1}^T(\hat\beta_{J_1})\bD_{J_1}(\beta^0(j),\hat\beta_{J_1})}n\\  &+ 
 \frac{\|\bD_{J_2}(\beta^0(j),\hat\beta_{J_2})\|_2^2}n -2\frac{\hat\be_{J_2}^T(\hat\beta_{J_2})\bD_{J_2}(\beta^0(j),\hat\beta_{J_2})}n\,.
\end{align*}
Then by Lemmas~\ref{lem2} and \ref{lem4} we have that
\begin{align}\label{bound-a}
L( J,\beta^0(j)) - &\,L(J_1,\hat\beta_{J_1}) - L(J_2,\hat\beta_{J_2}) \notag\\
 \leq&\;\, \|\bX_{J_1}(\beta^0(j) - \hat\beta_{J_1})\|_2^2/n +\lambda\sqrt{|J_1|} \|\beta^0(j) - \hat\beta_{J_1}\|_1\notag\\
& +\;  \|\bX_{J_2}(\beta^0(j) - \hat\beta_{J_2})\|_2^2/n +\lambda\sqrt{|J_2|} \|\beta^0(j) - \hat\beta_{J_2}\|_1\notag\\
\leq& \;\, 2c_*\lambda^2s_*\,.
\end{align}
Also by Lemma~\ref{lem4} we have that 
\[
\lambda\sqrt{|J|} \|\beta^0(j)-\hat\beta_J\|_1 \;\leq\; c(2)\lambda^2s_* \;\leq\; 4c_*\lambda^2s_* 
\]
therefore
\begin{align*}\label{bound1}
G(\alpha) - G(\hat\alpha) \;&\leq\;  6 c_* \lambda^2s_* -  \gamma
\end{align*}
and  if 
\[
\gamma >  6 c_* \lambda^2 s_*
\]
we obtain $G(\alpha)<G(\hat\alpha)$ which is a contradiction. \\
In  case (b), let $j_1$ be such that $\hat\alpha\cap \B(\alpha^0_{j_1},\lambda\sqrt\delta/d_*) = \emptyset$. We have $1\leq j_1 \leq k_0-1$. We now distinguish two possible sub-cases: 
(b1) $\hat\alpha\cap \B(\alpha^0_{j_1},\delta) = \emptyset$;  and (b2) $\hat\alpha\cap \B(\alpha^0_{j_1},\delta) \neq \emptyset$. In case (b1) we define $\alpha=\hat\alpha\cup\{\alpha^0_{j_1}\}$ and then $\alpha$ is a valid candidate vector for the minimization \eqref{hatalpha} because $r(\alpha)\geq\delta$. Denote by $J_1$ and $J_2$ the intervals in $\alpha$ that contain (as an extreme) the point $\alpha^0_{j_1}$; that is 
\[
J_1 = (\alpha_{r_i-1},\alpha^0_{j_1}]\,,\quad J_2 = (\alpha^0_{j_1},\alpha_{r_i+1}]\,,\qquad \alpha_{r_i} = \alpha^0_{j_1}\,,
\]
and let $J$ denote their union $J=(\alpha_{r_i-1},\alpha_{r_i+1}]$.  We have that
\begin{align}\label{Gdif}
G(\hat\alpha)&\, - G(\alpha)
\,=\; L( J,\hat\beta_{J}) - L(J_1,\hat\beta_{J_1})-L(J_2,\hat\beta_{J_2}) - \gamma\,.
\end{align}
By the equality \eqref{eqmain} (with $\eta=v$, $\beta=\hat\beta_J$ and $\beta^{(1)}=\beta^*_J$), Lemma~\ref{lem2} 
and Theorem~\ref{mixture}  we obtain 
\[
|L(J,\hat\beta_{J}) -  L(J,\beta^*_{J}) | \;\leq\;  \frac{8\lambda^2 s_*}{\phi_*^2}
\]
and the same applies to the intervals $J_1$ and $J_2$. 
Then, one more time by the equality \eqref{eqmain}  (with $\beta=\beta^*_J$, $\beta^{(1)}=\beta^*_{J_1}$ and $\beta^{(2)}=\beta^*_{J_2}$),  Lemmas~\ref{lem1} and \ref{lem2} and Theorem~\ref{mixture} 
we have that
\begin{align*}
L( J,\hat\beta_{J}) - L(J_1,\hat\beta_{J_1}) - L(J_2,\hat\beta_{J_2})\;&\geq \;  L( J,\beta^*_{J})  -  L(J_1,\beta^*_{J_1}) - L(J_2,\beta^*_{J_2}) - \frac{24\lambda^2s_*}{\phi_*^2} \\
&\geq\; \frac{\delta m_*^2 \phi_*^2 s_*}{8} - \lambda\sqrt\delta M_*s_* - \frac{24\lambda^2 s_*}{\phi_*^2} 
\end{align*}
and therefore, as $\lambda < \sqrt\delta M_*\phi_*^2/24$ we obtain  
\[
G(\hat\alpha) - G(\alpha) \;\geq\;  \frac{ \delta m_*^2 \phi_*^2 s_*}{8} - 2\lambda\sqrt\delta M_* s_*- \gamma\,.
\]
In this way, if
\[
\gamma + 2\lambda\sqrt\delta M_* s_* \;<\; \frac{m_*^2\phi_*^2 s_*}{8}\, \delta  \;=\; 4 d_*M_*s_*\delta
\]
then \eqref{ineq} is satisfied, contradicting the fact that $\hat\alpha$ is the minimizer of \eqref{hatalpha}. \\
For case (b2), a more elaborated argument is necessary, because if we add some of the points $\alpha^0_{j_i}$ to $\hat\alpha$ we obtain vectors with intervals of length smaller than $\delta$. Then we need to add some points and to remove others in order to obtain a good candidate vector. 
Define the vector $\alpha^{(1)} =  \hat\alpha\cup\{\alpha^0_{j_1}\}$. 
As before denote by $J_1$ and $J_2$ the intervals in $\alpha^{(1)}$ that contain (as an extreme) the point $\alpha^0_{j_1}$; that is 
\[
J_1 = (\alpha^{(1)}_{r_i-1},\alpha^0_{j_1}]\,,\quad J_2 = (\alpha^0_{j_1},\alpha^{(1)}_{r_i+1}]\,,\qquad \alpha^{(1)}_{r_i} = \alpha^0_{j_1}\,,
\]
and let $J$ denote their union $J=(\alpha^{(1)}_{r_i-1},\alpha^{(1)}_{r_i+1}]$.  By using the extended definition of $G$ in \eqref{G2} we have that
\begin{align}\label{Gdif}
G(\hat\alpha) - G(\alpha^{(1)})
\,=\; L( J,\hat\beta_{J}) - L(J_1,\hat\beta_{J_1})-L(J_2,\hat\beta_{J_2}) - \gamma\,.
\end{align}
If $|J_1|<\delta$,  by the condition $\delta+\lambda\sqrt\delta/d_*<r(\alpha^0)$ we must have $r_i\geq 2$ and there must exist an interval
$K_1=(\alpha^{(1)}_{r_i-2},\alpha^{(1)}_{r_i-1}]$ in $\alpha^{(1)}$ (adjacent to $J^i_1$ to the left), see Figure~\ref{intervals}.
Similarly for the interval  $J^i_2$, if $|J^i_2|<\delta$ then we must have  $r_i\leq \ell(\alpha^{(1)})-2$ and there must exist an interval
$K^i_2=(\alpha^{(1)}_{r_i+1},\alpha^{(1)}_{r_i+2}]$ in $\alpha^{(1)}$ (adjacent to $J^i_2$ to the  right). To take only one case from now on we assume $|J_1|<\delta$  and $|J_2|\geq\delta$, the other possibilities can be handled in a similar way. 
\begin{figure}[h!]
\includegraphics{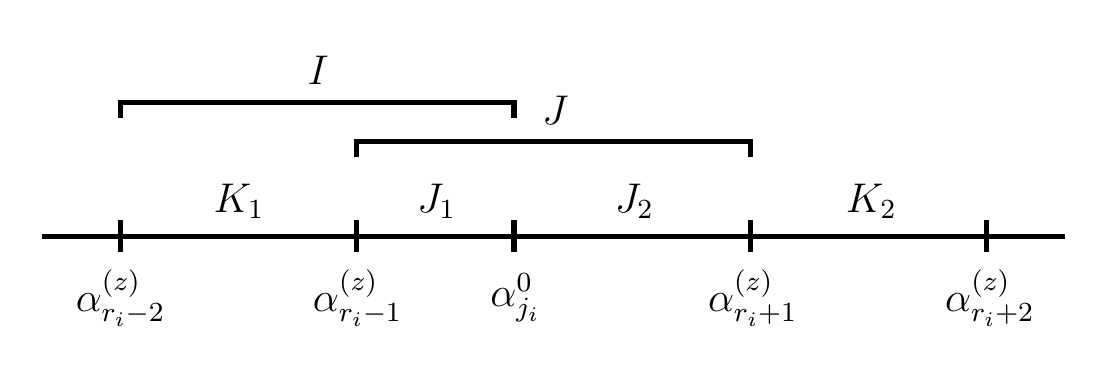}
\caption{Graphical representation of the intervals $J_1$, $J_2$, $K_1$, $K_2$, $I$ and $J$.}\label{intervals}
\end{figure}

Now, we will construct a vector $\alpha$ obtained from $\alpha^{(1)}$ by removing the component  $\alpha^{(1)}_{r_i-1}$, that is $\alpha=\alpha^{(1)}\setminus\{\alpha^{(1)}_{r_i-1}\}$.
In this case, by the  definition of the intervals $J_1$, $J_2$, $K_1$, $K_2$ and taking $I= K_1\cup J_1$ (see Figure~\ref{intervals}) we have that 
 \begin{align*}
G(\alpha^{(1)}) - G(\alpha) =\; L(K_1,\hat\beta_{K_1}) + L(J_1,\hat\beta_{J_1})  + \gamma - L(I,\hat\beta_{I})\,.
\end{align*}
therefore 
 \begin{align*}
G(\hat \alpha) - G(\alpha) & = G(\hat \alpha) - G(\alpha^{(1)}) +  
G(\alpha^{(1)}) - G(\alpha) \\
& =\; L( J,\hat\beta_{J}) - L(J_2,\hat\beta_{J_2})  +  L(K_1,\hat\beta_{K_1})  - L(I,\hat\beta_{I})\,.
\end{align*}
By using the same arguments and in case (b1), with the observation that $|I|\geq \delta$ implies
\begin{align*}
\sqrt{|I|}\|\beta^*_{K_1}- \hat\beta_{I}\|_1 &\leq \sqrt{|I|}\|\beta^*_{I}- \hat\beta_{I}\|_1 + \sqrt{|I|}\|\beta^*_{K_1}- \beta^*_{I}\|_1\\
&\leq  \sqrt{|I|}\|\beta^*_{I}- \hat\beta_{I}\|_1 + \sqrt{|I|}\frac{\delta}{|I|} M_* s_*\\
&\leq \sqrt{|I|}\|\beta^*_{I}- \hat\beta_{I}\|_1 + \sqrt\delta M_* s_*
\end{align*}
we have that
 \begin{align*}
L(J,\hat\beta_{J}) - & \,L(J_2,\hat\beta_{J_2}) + L(K_1,\hat\beta_{K_1}) - L(I,\hat\beta_{I})\\
&\geq \; L(J,\beta^*_{J})  - L(J_2,\beta^*_{J_2}) +   L(K_1,\beta^*_{K_1})  - \frac{24\lambda^2 s_*}{\phi_*^2}  \\
&\qquad  - L(I,\beta^*_{K_1}) - \lambda\sqrt{|I|}\|\beta^*_{K_1}- \hat\beta_{I}\|_1\\
&= \; L(J,\beta^*_{J})  - L(J_2,\beta^*_{J_2}) - L(J_1,\beta^*_{J_1})  - \frac{24\lambda^2 s_*}{\phi_*^2} 
\\
&\qquad -  \lambda\sqrt{|I|}\|\beta^*_{I}- \hat\beta_{I}\|_1 - \lambda\sqrt\delta M_*s_*  \\
&\geq \; \frac{\lambda\sqrt\delta m_*^2\phi_*^2 }{8 d_*}- 2\lambda\sqrt\delta M_*  - \frac{32\lambda^2}{\phi_*^2}\\
&= \; 2\lambda\sqrt\delta M_* - \frac{32\lambda^2}{\phi_*^2}  \,. 
\end{align*}
By the condition
\[
\lambda \;<\; \frac{\sqrt\delta M_* \phi_*^2}{24}
\]
we obtain 
\begin{equation*}
G(\hat\alpha) - G(\alpha) \;>\;  0\,,
\end{equation*}
contradicting the fact that $\hat\alpha$ minimizes \eqref{hatalpha}.
The last point in the theorem can be derived directly from Lemma~\ref{lem4} and $\|\hat\alpha-\alpha^0\|_1\leq \lambda\sqrt\delta/d_*$. 
\end{proof}

\begin{proof}[Proof of Theorem~\ref{modsel2}]
First we will show that under the conditions of Theorem~\ref{modsel}, on $\T_0\cap\T_1$ we have that $h(0,1)=0$ if $k_0=1$ or  $h(0,1)$ is at most at distance $\lambda\sqrt\delta/d_*$ of some of the values in $(\alpha^0_1,\dotsc,\alpha^0_{k_0-1})$ if $k_0>1$. This fact can be derived straightforward from the proof of Theorem~\ref{modsel}, as 
 the objective functions coincide for 1 or 2 segments; that is 
 \[
G((0,u,1]) = H(0,u) + H(u,1) \text{ for all }u\in [0,1]\,,
 \] 
where $G$ is given by  \eqref{G2} and $H$ is defined in \eqref{defH}. 

So first suppose $k_0=1$ and $\alpha^0=(0,1)$. Then  by the same arguments used in the proof of case (a) in Theorem~\ref{modsel} we have that for $\alpha(u)=(0,u,1)$ we must have
\[
G(\alpha^0) <  \min_{u\in(\delta,1-\delta)} G(\alpha(u)) 
\]
and therefore $h(0,1)=0$. 
Now suppose $k_0>1$ and that 
$h(0,1) \notin  \cup_{j=1}^{k_0-1}\B(\alpha^0_j,\lambda\sqrt\delta/d_*)$, 
define
\begin{align*}
\alpha^{(0)} & =  (0,h(0,1),1] \\
\alpha^{(1)} & = \alpha^{(0)} \cup \{\alpha^0_j\}\\
\alpha^{(2)} & = \alpha^{(1)}\setminus\{h(0,1)\}\,.
\end{align*}
If $h(0,1)=0$ (meaning that $\ell(\hat\alpha^{bs})=1$) we can apply the arguments of case (b1) in Theorem~\ref{modsel}, obtaining that $G(\alpha^{(0)})   -  G(\alpha^{(1)})>0$. On the other hand, if 
$h(0,1)\in [\delta,1-\delta]$ we can apply the same argument of case (b2) in Theorem~\ref{modsel}, obtaining  
\begin{align*}
G(\alpha^{(0)})   -  G(\alpha^{(2)})  \;=\; G(\alpha^{(0)})   -  G(\alpha^{(1)})  + G(\alpha^{(1)})   -  G(\alpha^{(2)})  \;>\; 0 \,.
\end{align*}
In both cases we contradict the fact that $h(0,1)$ minimizes \eqref{h}. So, if $k_0>1$ we must have $h(0,1) \in  \cup_{j=1}^{k_0-1}\B(\alpha^0_j,\lambda\sqrt\delta/d_*)$, with $j=1,\dotsc, k_0-1$. Now we can replicate the same argument above on each one of the sub-intervals $(0,h(0,1)]$ and  $(h(0,1),1]$ 
provided that $\delta  \leq r(\alpha^0) -  2 \lambda\sqrt\delta/d_*$.
\end{proof}

\section{Auxiliary results}

Given an interval $(u,v]\subset[0,1]$ denote by
\[
\gamma_{(u,v]}(\alpha^0_j)= |(u,v]\cap(\alpha^0_{j-1},\alpha^0_j]|/(v-u)\,,\qquad j=1,\dotsc, k_0\,,
\]
where $|(u,v]\cap(\alpha^0_{j-1},\alpha^0_j]|$ equals the length 
of the interval  $(u,v]\cap(\alpha^0_{j-1},\alpha^0_j]$. 

\begin{lemma}\label{oracle}
If $\Sigma$ is positive definite, 
for any interval $(u,v]\in I^n_0$ we have that
\begin{equation*}
\beta^*_{(u,v]}\;=\;  \sum_{j=1}^{k_0}  \gamma_{(u,v]}(\alpha^0_j)\,\beta^0(j) \,.
\end{equation*}
\end{lemma}

\begin{proof}
Observe that for $i \in (\alpha^0_{j-1}n,\alpha^0_j n]\cap(un,vn]$ we have 
\begin{align*}
\E |Y_i- X_i^T\beta|^2 &\;=\; \E |X_i^T(\beta^0(j)-\beta) +\epsilon_i |^2 \\
&\;=\;  \E |X_i^T(\beta^0(j)-\beta)|^2  + \E(\epsilon_i^2)\,.
\end{align*}
Therefore
\begin{align*}
\beta^*_{(u,v]} \;&=\; \underset{\beta}{\arg\min} \;\; n^{-1}\,\sum_{i=un+1}^{vn}  \E |X_i^T(\beta^{(i)}-\beta)|^2  \\
&=\; \underset{\beta}{\arg\min} \; \sum_{j=1}^{k_0} |(u,v]\cap(\alpha^0_{j-1},\alpha^0_j]|\,  (\beta^0(j)-\beta)^T \,\Sigma\,(\beta^0(j)-\beta) \\
&=\; \underset{\beta}{\arg\min} \; \beta^T \,\Sigma\, \beta - 2\beta^T\, \Sigma\,\Bigl(  \,\sum_{j=1}^{k_0}  \gamma_{(u,v]}(\alpha_j^0) \beta^0(j) \,\Bigr)\\
&=\; \underset{\beta}{\arg\min} \; (\beta-\tilde\beta)^T\, \Sigma \,(\beta -\tilde\beta)\,,
\end{align*}
where 
\[
\tilde\beta \;=\; \sum_{j=1}^{k_0}  \gamma_{(u,v]}(\alpha_j^0) \beta^0(j)\,.
\]
If $\Sigma$ is positive definite we have that the minimizer  is $\beta^*_{(u,v]} = \tilde\beta$ and this concludes the proof of Proposition~\ref{oracle}. 
\end{proof}

We now prove a basic result about the constant $m_*$ defined by Assumption~4.

\begin{lemma}\label{convex}
If Assumption~4 holds then 
\[
\inf_{j=1,\dotsc, k_0-1} \,\inf_{(u,\alpha^0_j],(\alpha^0_j,v]  \in I^n_0} \|\beta^*_{(u,\alpha^0_j]}- \beta^*_{(\alpha^0_j,v]}\|_1 \;\geq \; m_*s_*\,.
\]
\end{lemma}

\begin{proof}
As the $\ell_1$-norm is a sum over the different coordinates, we will minimize over each coordinate separately. 
So, fix $j=1,\dotsc, k_0$  and $i=1,\dotsc, p$;  we will show that for any $y_i\in\R$ (fixed), the minimizer of 
\[
| y_i  - (\beta^*_{(\alpha^0_j,v]})_i |
\]
over the set $\{v\colon \alpha^0_j\leq v\leq 1\}$ is one of the $(\beta^*_{(\alpha^0_j,\alpha^0_k]})_i$, with $k=j+1,\dots, k_0$.
But this is equivalent to the following optimization problem
\begin{align*}
\text{Minimize:}& \quad \sum_{k=j}^{k_0-1} \frac{ \max(v-\alpha^0_k,0)}{(v-\alpha^0_j)} 
 (y_i - \beta_i^0(k+1)) \\
\text{Subject to:}& \quad \alpha^0_j \leq v \leq \alpha^0_{k_0} = 1\,,
\end{align*}
where the objective function is continuous and linear on each of the intervals 
$(\alpha^0_{k},\alpha^0_{k+1}]$ for $k=j,\dotsc, k_0-1$. 
Therefore the solution  must be attained at one of the ``vertices''
$\{\alpha^0_k\colon k=j+1,\dotsc,k_0\}$. The same result can be obtained
fixing $v$ and minimizing over $u < \alpha^0_j$, then the statement of the
lemma follows. 
\end{proof}
 
\begin{lemma}\label{lemma-t0}
Suppose Assumptions 1 and 2 hold. Then for all $t>0$ and
\[
\lambda_0 =  14 t \sigma K_X \sqrt{\frac{\log(n^2p)}{n}}
\]
we have 
\[
\P(\T_0) \geq 1- 1/t^2\,.
\]
\end{lemma}

\begin{proof}
For any $i=1,\dotsc, n$ define the vector $\bZ_i\in \R^d$, with $d=p n(n-1)$ as 
\begin{equation}
(\bZ_i)_{j,(u,v]}  = \begin{cases}
\epsilon^{*}_{(u,v],i} X_{i}^{(j)} - \E(\epsilon^{*}_{(u,v],i}X_i^{(j)})\,; & \text{ if }i/n \in (u,v]\,,\\
0\,; & \text{ c.c. }
\end{cases}
\end{equation}
We have that $\bZ_1,\dotsc, \bZ_n$ are independent, with $\E(\bZ_i) = 0$
for all $i$. By Assumptions 1 and 2  we also have that $\E\|\bZ_i\|_\infty^2\leq 4\sigma^2 K_X^2$ for all $i$.
Denote by $\bS_n = \sum_{i=1}^n \bZ_i$. 
By Markov's inequality  we have that  
\begin{align*}
\P\Bigl( \, 2\|n^{-1}\bS_n\|_\infty \,>\, \lambda_0 \Bigr)
&\;\leq\;  \frac{4\,\E   \|n^{-1/2}\bS_n\|_\infty^2}{n\lambda_0^2}\,.
\end{align*}
Now, by \cite[Corollary~2.3]{dumbgen2010} and Assumptions 1 and 2 we have
that
\begin{align*}
\E   \|n^{-1/2}\bS_n\|_\infty^2\,&\leq \, n^{-1}(2e\log d - e)\sum_{i=1}^n\|\bZ_i\|_\infty^2\\
&\leq \, 8e K_X^2\sigma^2 \log d
\end{align*}
therefore
\begin{align*}
\P\Bigl( \, 2\|n^{-1}\bS_n\|_\infty \,>\, \lambda_0 \Bigr)
&\;\leq\;  \frac{ 32eK_X^2 \sigma^2\log(n^2p)}{n \lambda_0^2}\,.
\end{align*}
Moreover, as $\epsilon^*_{(u,v]}$ is orthogonal to $\bX^{(j)}_{(u,v]}$ in the $L^2(P)$ space for all  $j=1,\dotsc,p$ and all $(u,v]\in I_0^n$ then  
\begin{equation*}
2\|n^{-1}\bS_n\|_\infty= \underset{(u,v]\in I_0^n}{\max}\;\underset{1\leq j\leq p}{\max}\; 2\bigl|  \be^{*,T}_{(u,v]}\bX_{(u,v]}^{(j)}\bigr|/n\,.
\end{equation*}
and this concludes the proof.
\end{proof}

Now let $\T_1$ be given by  
\begin{equation*}\label{t1}
\T_1 = \Bigl\{ \underset{(u,v]\in I_0^n}{\max} \,\| \hat\Sigma_{(u,v]}-(v-u)\Sigma\|_\infty \leq 
\lambda_1  \Bigr\} \,,
\end{equation*}
where 
\[
\hat\Sigma_{(u,v]} = \bX_{(u,v]}^T\bX_{(u,v]}/{n}\,.
\]

\begin{lemma}\label{lemma-t1}
If Assumption 1 holds then for all $t>0$ and
\[
\lambda_1 = 10 t  K_X^2 \sqrt{\frac{2\log(np)}{n}}
\]
we have 
\[
\P(\T_1)\;\geq\; 1- 1/t^2\,.
\]
\end{lemma}

\begin{proof}
For 
any $i=1,\dotsc, n$ define the vector $\bW_i\in \R^{d}$, with $d=p^2n(n-1)$ as 
\begin{equation*}
(\bW_i)_{j,l,(u,v]}  = \begin{cases}
X_i^{(j)}X_{i}^{(l)} - \E(X_i^{(j)}X_i^{(l)})\,; & \text{ if }(u,v]\in I^n_0,\,i/n \in (u,v]\,,\\
0\,; & \text{ c.c. }
\end{cases}
\end{equation*}
We have that $\bW_1,\dotsc, \bW_n$ are independent, with $\E(\bW_i) = 0$ for all $i$. By Assumption 1 we also have that $\E\|\bW_i\|_\infty^2\leq 4K_X^4$ for all $i$.
Denote by $\bS_n = \sum_{i=1}^n \bW_i$. 
By Markov's inequality  we have that  
\begin{align*}
\P\Bigl( \, \|n^{-1}\bS_n\|_\infty \,>\, \lambda_1 \Bigr)
&\;\leq\;  \frac{\E   \|n^{-1/2}\bS_n\|_\infty^2}{n\lambda_1^2}\,.
\end{align*}
Now, by \cite[Corollary~2.3]{dumbgen2010}  we have
that
\begin{align*}
\E   \|n^{-1/2}\bS_n\|_\infty^2\,&\leq \, n^{-1}(2e\log d - e)\sum_{i=1}^n\|\bW_i\|_\infty^2\\
&\leq \, 8e K_X^4\log d
\end{align*}
therefore
\begin{align*}
\P\Bigl( \,\|n^{-1}\bS_n\|_\infty \,>\, \lambda_1 \Bigr)
&\;\leq\;  \frac{ 32eK_X^4\log (n^2p^2)}{n\lambda_1^2}\,.
\end{align*}
The proof finishes by noting that 
\begin{equation*}
\|n^{-1}\bS_n\|_\infty= \underset{(u,v]\in I_0^n}{\max}\;\underset{1\leq j,\,l\leq p}{\max}\; \bigl| \bX_{(u,v]}^{(j)}\bX_{(u,v]}^{(l)}/n - (v-u)\Sigma_{j,l}   \bigr|\qedhere
\end{equation*}
\end{proof}

\bibliographystyle{plain}
\bibliography{./references3}

\end{document}